%% file: main_COAP.tex
\documentclass[sn-mathphys,Numbered]{sn-jnl}
 
\usepackage{amssymb,amsmath,amsfonts,mathrsfs,amsthm}
\usepackage[utf8]{inputenc}
\usepackage[T1]{fontenc}
\usepackage{caption}
\usepackage[english]{babel}
\usepackage{graphicx}
\usepackage[title]{appendix}
\usepackage{color}
\usepackage{doi}
\usepackage{commath}
\usepackage{array}
\usepackage[normalem]{ulem}
\usepackage{multirow}%
\usepackage{mathrsfs}%
\usepackage{xcolor}%
\usepackage{textcomp}%
\usepackage{manyfoot}%
\usepackage{booktabs}%

\theoremstyle{thmstyleone}%
\newtheorem{theorem}{Theorem}%  meant for continuous numbers
\newtheorem{lemma}{Lemma}% 
\newtheorem{corollary}{Corollary}
\newtheorem{proposition}{Proposition}% to get separate numbers for theorem and proposition etc.

\theoremstyle{thmstyletwo}%
\newtheorem{remark}{Remark}%

\theoremstyle{thmstylethree}%
\newtheorem{definition}{Definition}%

\input{grib_defs}

\DeclareMathOperator{\QEnum}{\EC_{\PC}}
\DeclareMathOperator{\ZEnum}{\overline{\QEnum}}

\begin{document}

\title[Faster Representation for Counting Integer Points \ldots]{A New and Faster Representation for Counting Integer Points in Parametric Polyhedra
}
\author*[1,2]{\fnm{Dmitry V.} \sur{Gribanov}}\email{dimitry.gribanov@gmail.com}

\author[2]{\fnm{Dmitry S.} \sur{Malyshev}}\email{dsmalyshev@rambler.ru}

\author[3]{\fnm{Panos M.} \sur{Pardalos}}\email{panos.pardalos@gmail.com}

\author[4]{\fnm{Nikolai Yu.} \sur{Zolotykh}}\email{nikolai.zolotykh@itmm.unn.ru}

\affil*[1]{\orgdiv{Laboratory of Discrete and Combinatorial Optimization}, \orgname{Moscow Institute of Physics and Technology}, \orgaddress{\street{Institutsky lane 9}, \city{Dolgoprudny, Moscow region}, \postcode{141700}, \country{Russian Federation}}}

\affil[2]{\orgdiv{Laboratory of Algorithms and Technologies for Network Analysis}, \orgname{HSE University}, \orgaddress{\street{136 Rodionova Ulitsa}, \city{Nizhny Novgorod}, \postcode{603093}, \country{Russian Federation}}}

\affil[3]{\orgdiv{Department of Industrial and Systems Engineering}, \orgname{University of Florida}, \orgaddress{\street{401 Weil Hall}, \city{Gainesville}, \postcode{116595}, \state{Florida}, \country{USA}}}

\affil[4]{\orgname{Lobachevsky State University of Nizhny Novgorod}, \orgaddress{\street{23 Gagarina Avenue}, \city{Nizhny Novgorod}, \postcode{603950}, \country{Russian Federation}}}

\abstract{
In this paper, we consider the counting function $\QEnum(y) = |\PC_{y} \cap \ZZ^{n_x}|$ for a parametric polyhedron $\PC_{y} = \{ x \in \RR^{n_x} \colon A x \leq b + B y\}$, where $y \in \RR^{n_y}$. We give a new representation of $\QEnum(y)$, called a \emph{piece-wise step-polynomial with periodic coefficients}, which is a generalization of piece-wise step-polynomials and integer/rational Ehrhart's quasi-polynomials. It gives the fastest way to calculate $\QEnum(y)$ in certain scenarios.

\medskip

\textbf{The most important cases are the following:}

1) We show that, for the parametric polyhedron $\PC_y$ defined by a standard-form system $A x = y,\, x \geq 0$ with a fixed number of equalities, the function $\QEnum(y)$ can be represented by a polynomial-time computable function. In turn, such a representation of $\QEnum(y)$ can be constructed by an $\poly\bigl(n, \|A\|_{\infty}\bigr)$-time algorithm;
\medskip

2) Assuming again that the number of equalities is fixed, we show that integer/rational Ehrhart's quasi-polynomials of a polytope can be computed by FPT-algorithms, parameterized by sub-determinants of $A$ or its elements; 
\medskip

3) Our representation of $\QEnum$ is more efficient than other known approaches, if $A$ has bounded elements, especially if it is sparse in addition;
\medskip

Additionally, we provide a discussion about possible applications in the area of compiler optimization. In some “natural” assumptions on a program code, our approach has the fastest complexity bounds.
}%abstract

\maketitle              % typeset the header of the contribution

\keywords{Integer Linear Programming, Parametric Integer Programming, Short Rational Generating Function, Bounded Sub-Determinants, Multidimensional Knapsack Problem, Subset-Sum Problem, Counting Problem}

% \tableofcontents

\section{Introduction}\label{intro_sec}

In our paper, we discuss different existing and new approaches for the problem to calculate the number of points with integer coordinates in a polyhedron, which is defined by a system of linear inequalities that additionally depends on a vector of parametric variables. Formal definitions of considered problems will be presented in the next Subsection. As far as the authors know, from the application point of view, this framework is used in the areas of creating intelligent systems of analysis, profiling, control and optimization of program code. In their pioneer works, Loechner \& Wilde \cite{VerticesOfParametricPolyhedra} and Clauss \& Loechner \cite{Clauss_EhrhartApp,Parametric_Clauss} present applications to automatic parallelization, estimation of a nested loop execution time, estimation of maximum parallelism, etc. We would also like to mention recent progress on the cache miss calculation, see Bao et al. \cite{AnalyticalСache2017}, Gysi et al. \cite{FastAssociativeCache2019}, Shah et al. \cite{BullsEye2022}. Our work is theoretical. We introduce a new representation of the counting function, called \emph{piece-wise periodic step-polynomials}, which is more efficient than other known approaches in certain scenarious. Additionally, we discuss its connections with existing representations, such as \emph{piece-wise Ehrhart's quasi-polynomials} and \emph{piece-wise step-polynomials}. In fact, the new representation generalizes both of them.

\subsection{Main problem statement}

Let $n_x$ and $n_y$ be the dimensions of $x$ and parametric $y$ variables, and let a polytope $\PC$ be defined by one of the following ways: 
\begin{enumerate}
\item[(i)] {\bf System in the canonical form:} 
\begin{equation}\tag{Canon-Form}\label{canonical_form}
    \PC = \left\{\binom{x}{y} \in \RR^{n_x+n_y} \colon A x \leq b + B y \right\},
\end{equation}
where $A \in \ZZ^{m \times n_x}$, $B \in \QQ^{m \times n_y}$, and $b \in \QQ^{m}$;
\item[(ii)] {\bf System in the standard form:} 
\begin{equation}\tag{Standard-Form}\label{standard_form}
    \PC = \left\{\binom{x}{y} \in \RR^{n_x + n_y} \colon A x = b + B y,\, x \geq 0\right\},
\end{equation}
where $A \in \ZZ^{k \times n_x}$, $B \in \QQ^{k \times n_y}$, and $b \in \QQ^{k}$.
\end{enumerate} 
We put
\begin{equation*}
    \PC_{y} = \left\{x \in \RR^{n_x} \colon \binom{x}{y} \in \PC\right\},
\end{equation*} and consider the counting function
\begin{gather*}
    \QEnum \colon \RR^{n_y} \to \ZZ_{\geq 0} \cup \{+\infty\},\qquad\text{given by}\notag\\
    \QEnum(y) = \abs{\PC_{y} \cap \ZZ^{n_x}},
\end{gather*}
and its restriction on $\ZZ^{n_y}$, denoted by:
\begin{equation*}
   \ZEnum = \restr{\QEnum}{\ZZ^{n_y}}.
\end{equation*}
{\bf Our paper is motivated by the following computational problem: } For the input $(A,B,b)$, construct an efficient representation of $\QEnum$, which will allow to calculate quickly the value of $\QEnum(y)$, for any $y \in \QQ^{n_y}$. By the word “efficient”, we mean that the function $\QEnum$ is encoded, using some non-trivial data structure that can faster perform queries to $\QEnum$ in comparison with approaches that have no prior information on $(A,B,b)$. 

We study the computational complexity of this problem with respect to several parameters of $A$. The first of them correspond to a sub-determinant structure of $A$.
\begin{definition}
For a matrix $A \in \ZZ^{m \times n}$, by $$
\Delta_k(A) = \max\left\{\abs{\det (A_{\IC \JC})} \colon \IC \subseteq \intint m,\; \JC \subseteq \intint n,\; \abs{\IC} = \abs{\JC} = k\right\},
$$ we denote the maximum absolute value of determinants of all the $k \times k$ sub-matrices of $A$. Here, the symbol $A_{\IC \JC}$ denotes the sub-matrix of $A$, which is generated by all the rows with indices in $\IC$ and all the columns with indices in $\JC$. Note that $\Delta_1(A) = \|A\|_{\max}$.  

By $\Delta_{\gcd}(A,k)$ and $\Delta_{\lcm}(A,k)$, we denote the greatest common divisor and least common multiplier of non-zero determinants of all the $k \times k$ sub-matrices of $A$. Additionally, let $\Delta(A) = \Delta_{\rank(A)}(A)$, $\Delta_{\gcd}(A) = \Delta_{\gcd}(A,\rank(A))$, and $\Delta_{\lcm}(A) = \Delta_{\lcm}(A,\rank(A))$. The matrix $A$ with $\Delta(A) \leq \Delta$, for some $\Delta > 0$, is called \emph{$\Delta$-modular}.
\end{definition}

\begin{definition}\label{codim_def}
For a system in \ref{standard_form}, we call $\rank(A)$ as the \emph{co-dimension} of $\PC$. In turn, for a system in \ref{canonical_form}, we define the \emph{co-dimension} to be equal $m - \rank(A)$.
\end{definition}
\begin{remark}\label{form_trans_rm}
    Note that this definition is very natural for systems in \ref{standard_form}. For systems in \ref{canonical_form}, the definition can be justified in the following way. Assume that $\rank(A) = n$, for \ref{canonical_form}, and $\rank(A) = k$, for \ref{standard_form}. Due to \cite[Lemma~4 and Lemma~5]{OnCanonicalProblems_Grib}, any system in \ref{standard_form} can be polynomially transformed to a system in \ref{canonical_form} with $m = n + k$, preserving $\dim(\PC)$ and $\Delta(A)$. Vise versa, any system in \ref{canonical_form} can be polynomially transformed to a system in \ref{standard_form} with $k = m - n$ and an additional group-like constraint. This transformation also preserves $\dim(\PC)$ and $\Delta(A)$. By this reason, for polyhedra, defined by \ref{canonical_form}, the value $m-n$ is also called the \emph{co-dimension} of $\PC$.    
\end{remark}

The next two sufficiently general matrix parameters, denoted by $\nu(A)$ and $\mu(A)$, that will be considered in the whole work are related to a structure of the fan, induced by $A^\top$.
\begin{definition}
For arbitrary $A \in \ZZ^{m \times n}$ and $b \in \QQ^m$, put
\begin{gather*}
    \MC(A,b)=\bigl\{x \in \RR^{n} \colon A x \leq b\bigr\},\\
    \quad\nu(A) = \max\limits_{b \in \QQ^m} \abs{\vertex\bigl(\MC(A,b)\bigr)}.
\end{gather*}
\end{definition}

\begin{definition}
For an arbitrary matrix $A \in \RR^{n \times m}$, the symbols $\cone(A)$ and $\inth(A)$ denote the \emph{cone} and \emph{lattice}, induced by columns of $A$, i.e.
\begin{gather*}
    \cone(A) = \bigl\{At \colon t \in \RR^m_{\geq 0}\bigr\},\\
    \inth(A) = \bigl\{At \colon t \in \ZZ^m \bigr\}.
\end{gather*}
    
\end{definition}

\begin{definition}\label{triang_def}
For an arbitrary matrix $A \in \RR^{n \times m}$ of rank $n$, we define the parameter $\mu(A)$ as the maximum size of a triangulation of $\cone(A)$ with \emph{simple} cones, where the cone $\CCal$ is called \emph{simple}, if it is induced by columns of some $n \times n$ non-singular sub-matrix of $A$. More formally, a set $\TS$ is a \emph{triangulation} of the cone $\CCal = \cone(A)$, if the following requirements hold:
\begin{enumerate}
    \item For any cone $\TC \in \TS$, $\TC = \cone(A_{\BC})$, where $\BC$ is some $n\times n$ base of $A$;
    \item The equality $\CCal = \bigcup\limits_{\TC \in \TS} \TC$ is true;
    \item For different $\TC_1, \TC_2 \in \TS$, the set $\TC_1 \cap \TC_2$ forms a face of both $\TC_1$ and $\TC_2$.
\end{enumerate} 

We denote $\mu(A)$ as $\max\bigl\{\abs{\TS} \colon \text{$\TS$ is a triangulation of $\cone(A)$}\bigr\}$. Note that $\nu(A) \leq \mu(A^\top)$.
\end{definition}

Throughout the paper, we will use the following short notations with respect to the definitions \ref{canonical_form} and \ref{standard_form}: $\Delta := \Delta(A)$, $\Delta_1 := \Delta_1(A)$, $\nu := \nu(A)$, and $\mu := \mu(A^\top)$. Additionally, for $k \in \intint[0]{n_x + n_y}$, we denote the number of $k$-faces of $\PC$ by the symbol $f_{k}$. In other words, the values $f_k$ form components of the \emph{$f$-vector} of $\PC$. 
When estimating the computational complexity of algorithms, we will often use the notion of an \emph{FPT-algorithm}.
\begin{definition}
    An algorithm, parameterized by a parameter $k$, is called \emph{fixed-parameter tractable} (or, simply, an FPT-\emph{algorithm}) if its computational complexity can be estimated by a function from the class $f(k)\cdot n^{O(1)}$, where $n$ is the input size and $f(k)$ is a computable function that depends on $k$ only.
    A computational problem, parameterized by a parameter $k$, is called \emph{fixed-parameter tractable} (or, simply, an FPT-\emph{problem}) if it can be solved by an FPT-algorithm. For more information about the parameterized complexity theory, see \cite{FPT_Downey,FPT_Fomin}.
\end{definition}

Our main contributions with respect to the parametric counting problem are emphasized in Section \ref{summary_sec}. But, it uses some notations, which will be introduced later.
\begin{remark}\label{complexity_rm}
To make the text easier to read, we hide multiplicative terms of the type $\poly(\phi)$, using the $O^*(\cdot)$-notation, when we estimate the computational complexity. Here, $\phi$ denotes the input size. For example, the equality $n_x^{O(1)} = O^*(1)$ holds. Similarly, we use the $\tilde O(\cdot)$-notation to hide logarithmic terms. 

The outputs and all intermediate variables, occurring in the proposed algorithms, have polynomial-bounded bit-encoding size. Hence, any algorithm that is polynomial-time in terms of the arithmetic complexity analysis is a polynomial-time algorithm in terms of the bit-complexity analysis.
\end{remark}

\subsection{Structure of this work}
In Subsection \nameref{nonparam_survey_subs}, we give a survey for the non-parametric ($n_y = 0$) counting problem. In Subsection \nameref{param_survey_subs}, we present an introduction to the general parametric counting problem and give a survey of the known results. In Section \nameref{new_repr_sec}, we present our main theoretical contribution: a new representation of the parametric counting function $\QEnum$, named the \emph{periodic piece-wise step-polynomial}. In turn, in Subsection \nameref{quasipoly_subs}, we describe connections of the new representation with rational/integer Ehrhart's quasi-polynomials of $\PC$ and show how these quasi-polynomials can be computed, using our new representation. In Section \nameref{summary_sec}, we describe the main implications of our work from the theoretical and computational perspectives. 
% In Subsection \nameref{about_subs}, we say a few words about our method, developed in the series of papers \cite{SparseILP_Gribanov,Counting_FPT_Delta_corrected,CountingFixedM,Counting_FPT_Delta}, and describe differences with the previous works in the series. 
In Section \nameref{special_sec}, we prove the main computational implications for the two most important cases: the polyhedra, defined by systems of a bounded co-dimension $k$ (see Definition \ref{codim_def}), the general type polyhedra with bounded dimension $n_y$ of the parametric space. They are considered in Subsections \nameref{codim_subs} and \nameref{nyfixed_subs},  respectively. In Section \nameref{prelim_sec}, we describe auxiliary definitions and facts from the polyhedral analysis, which are necessary to prove the main theorems. More precisely, Subsection \nameref{valuations_subs} presents an introduction to the theory of valuations, indicator and generating functions on polyhedra; Subsection \nameref{vdtt_subs} presents important facts on vertices, edges, tangent cones, and triangulations; Subsection \nameref{unbouded_subs} proves auxiliary lemmas that help to handle the cases, when $\PC_y$ is unbounded, for some $y \in \RR^{n_y}$. Finally, Sections \nameref{chamber_decomp_th} and \nameref{main_param_th_proof} give proofs of the main Theorems \ref{main_param_th} and \ref{chamber_decomp_th} of our work, respectively.

\subsection{Survey on the non-parametric case}\label{nonparam_survey_subs}
Let us first survey known results about the non-parametric case, i.e. $n_y = 0$, and denote $n := n_x$ in the remaining part of this Subsection. The asymptotically fastest algorithm for the counting problem in a fixed dimension can be obtained, using the approach of A.~Barvinok \cite{Barv_Original} with modifications, due to Dyer \& Kannan \cite{OnBarvinoksAlg_Dyer} and Barvinok \& Pommersheim \cite{BarvPom}. A complete exposition of the Barvinok's approach can be found in \cite{BarvBook,BarvPom,BarvWoods,continuous_discretely,AlgebracILP}, additional discussion with respect to the dual-type counting algorithms can be found in the book \cite{counting_Lasserre_book}, due to J.~Lasserre. An important notion of the \emph{half-open sign decomposition} and other variant of Barvinok's algorithm that is more efficient in practice is given by K\"oppe \& Verdoolaege in \cite{HalfOpen}. Barvinok \& Woods \cite{BarvWoods} give important generalizations of the original techniques and adapts them to a wider range of problems to handle projections of polytopes.
Using the fastest deterministic Shortest Lattice Vector Problem (SVP) solver by Micciancio \& Voulgaris \cite{SVP_exp}, Barvinok's algorithm computational complexity can be estimated by
\begin{equation}\label{BarvComplexity}
\nu \cdot 2^{O(d)} \cdot \bigl( \log_2(\Delta) \bigr)^{d \ln(d)},    
\end{equation}
where $d := \dim(\PC)$. Since any polytope can be transformed to an integer-equivalent simple polytope, using a slight perturbation of the r.h.s. vector $b$, the parameterization by $\nu$ is correct (see, for example, \cite[Theorem~3]{Counting_FPT_Delta}). Since, for a fixed $d$, the value of $\nu$ and the value of $O\bigl( \log_2(\Delta) \bigr)^{d \log d}$ are bounded by a polynomial on the input length, the Barvinok's work shows that the counting problem is polynomial-time solvable in a fixed dimension. 

Lasserre \& Zeron \cite{simple_formula_counting} give formulae, based on the R.~Gomory's group-theoretic approach, whose complexity could be roughly bounded by
\begin{equation}\label{lasserre_complexity2}
\nu \cdot d^{O(1)} \cdot \Delta^{d}.    
\end{equation}
For $\Delta = O(d)$, the last complexity bound is better than \eqref{BarvComplexity}.
The papers \cite{SparseILP_Gribanov,CountingFixedM,Counting_FPT_Delta} are aimed to develop a counting algorithm with the $\poly(\nu,d,\Delta)$ complexity bound. Due to \cite{SparseILP_Gribanov}, the state of the art bound is
\begin{equation}\label{Delta_nu_complexity_result}
    O(\nu^2 \cdot d^4 \cdot \Delta^3). 
\end{equation}
Using the last complexity bound \eqref{Delta_nu_complexity_result} and different ways to estimate the parameter $\nu$, the papers \cite{SparseILP_Gribanov} and \cite{Counting_FPT_Delta} give new interesting complexity bounds for the $\Delta$-modular ILP feasibility problem, multi-dimensional knapsack problem, sparse ILP problems, and combinatorial multi-cover/multi-packing problems on hypergraphs. Table \ref{tab:comparison_primal} gives a comparison of the considered algorithms.

\begin{table}[h!]
    \centering
    \begin{tabular}{||c|c||}
    \hline
        $\nu \cdot 2^{O(d)} \cdot \bigl(\log_2(\Delta)\bigr)^{d \ln(d)}$ & see \cite[Chapter~16]{BarvBook} plus \cite{SVP_exp} \\
        \hline
         $\nu \cdot d^{O(1)} \cdot \Delta^{d}$ & see \cite{simple_formula_counting} \\
         \hline
         $O\bigl( \nu \cdot d^2 \cdot d^{\log_2(\Delta)} \bigr)$ & see \cite{CountingFixedM} \\
         \hline
         $O\bigl(\nu^2 \cdot d^4 \cdot \Delta^4 \cdot \log(\Delta) \bigr)$ & see \cite{Counting_FPT_Delta_corrected} \\
         \hline
         $O\bigl(\nu^2 \cdot d^4 \cdot \Delta^3 \bigr)$ & see \cite{SparseILP_Gribanov} \\
         \hline
    \end{tabular}
    \caption{Comparison of different primal-type algorithms}
    \label{tab:comparison_primal}
\end{table}

{\bf The case of a bounded co-dimension.} Consider now the polytopes defined by systems in the \ref{standard_form} with a bounded co-dimension $k$ (the number $k$ of linear independent rows in $A$ is bounded, see Definition \ref{codim_def}). The next natural question is the following: is it possible to compute $|\PC \cap \ZZ^n|$ by an FPT-algorithm with respect to the parameters $k$ \& $\Delta$ or $k$ \& $\Delta_1$, where $k$ is the co-dimension of $\PC$? The paper \cite{Counting_FPT_Delta} with a modification from \cite{SparseILP_Gribanov} gives a partially positive answer on this question. More precisely, for any fixed $k$, the problem to find $|\PC \cap \ZZ^n|$ can be solved by an FPT-algorithm, parameterized by $\Delta$ with the arithmetic complexity bound $O(n/k)^{2k+4} \cdot \Delta^3$. A similar parameterized algorithm with respect to $\Delta_1$ can be achieved just by using the Hadamard's bound. For $k=1$, it gives an $O(n^6 \cdot \Delta_1^3 )$ FPT-algorithm to count the solutions of the unbounded subset-sum problem. For $k=1$, the result of \cite{Counting_FPT_Delta} is not new, the earlier paper \cite{knapsack_lasserre}, due to Lasserre \& Zeron, also gives a dual-type FPT-algorithm for this problem, but the concrete complexity bound was not given. Finally, Lasserre \& Zeron \cite{fixedM_counting_lasserre} give a dual-type algorithm that is designed for systems with small values of $k$ and $\Delta_{total}(A)$ parameters. Unfortunately, the computational complexity analysis is not completely finished.

The ideas of dual-type algorithms and its residue techniques have been improved in the papers \cite{knapsack_path_int,IP_complex_int_SP,CountingViaContourIntegration}. Due to Hirai, Oshiro \& Tanaka \cite{CountingViaContourIntegration}, dual-type algorithms can be significantly more memory-saving than primal type-algorithms. For example, Hirai et al. prove the existence of an $O\bigl(\|y\|^k_{\infty}\bigr)$-time and $\poly(n,k,\|y\|_{\infty})$-space counting algorithm. Here, it is additionally assumed that $A$ is non-negative and $y \in \ZZ^k_{\geq 0}$.

\subsection{Introduction to the general parametric case}\label{param_survey_subs}
Let us return to the parametric case. Consider first the polytope $\PC = \{x \in \RR^{n_x} \colon A x \leq y \cdot b\}$, for $y \in \ZZ_{>0}$. In other words, $\PC_y = y \cdot \PC_1$. It was shown by Ehrhart \cite{Ehrhart1,Ehrhart2} that $\ZEnum(y)$ can be represented by a univariate polynomial with periodic coefficients, which is known as the  \emph{Ehrhart's quasi-polynomial}.
The least common multiple of all the coefficient periods is bounded by $t$, where $t \in \ZZ_{>0}$ is the minimum value, such that $\PC_t = t \cdot \PC_1$ becomes a polyhedron with $\vertex(\PC) \subseteq \ZZ^{n_x}$.  

It is difficult to directly store the Ehrhart's quasi-polynomial representation of $\QEnum(y)$, because it needs $O(n_x \cdot t)$ space with $t = 2^{O(\phi)}$, where $\phi$ is the bit-encoding length of $(A,b)$. However, it was shown by Barvinok \cite{Simplex_Barv} that the values of the first $j$ leading coefficients of the Ehrhart's quasi-polynomial of a rational simplex in a given point $y$ can be computed by a polynomial-time algorithm, assuming that $j$ is fixed. Moreover, this result holds even for more general polytopes, which have a fixed co-dimension, see Definition \ref{codim_def}. As it was noted in \cite{Simplex_Barv}, if the dimension is fixed in advance, the value of any periodic coefficient in a given point $y$ can be computed by a polynomial-time algorithm, using the interpolation technique.

The multivariate generalization of the above result, due to Ehrhart, was presented by Clauss, Loechner \& Wilde in \cite{Parametric_Clauss,VerticesOfParametricPolyhedra}. In order to give a formal exposition, we need to make a few definitions, which will be used significantly in the further text. 
\begin{definition}\label{perfun_def}
For $\DC \subseteq \RR^n$, we call a function $f \colon \DC \to \RR$ \emph{periodic}, if there exists a matrix $P$ with linear independent columns, such that $f(x) = f(y)$, whenever $x-y \in \inth(P)$. The matrix $P$ is called a \emph{period-matrix} of $f$. A vector $p$ is called a \emph{period-vector} or a \emph{multi-period} of $f$, if $\diag(p)$ is a \emph{period-matrix} of $f$. That is, $f(x) = f(y)$, whenever, for each $i \in \intint{n}$, $x_i = y_i + p_i \cdot k$, and $k \in \ZZ$.

We call a periodic function $f$ \emph{rational}, if its restriction $\restr{f}{\QQ^n}$ is periodic (the restriction has a rational period-matrix). Similarly, if $\restr{f}{\ZZ^n}$ is periodic (the restriction has an integer period-matrix), then $f$ is called \emph{integer}.
\end{definition}

\begin{definition}\label{quasipoly_def}
A \emph{quasi-polynomial} of degree $d$ in $n$ variables $x$ is a polynomial expression of degree $d$ in $x$ with coefficients, represented by periodic functions. That is,
$$
f(x) = \sum\limits_{\substack{\jB \in \ZZ_{\geq 0}^{n}\\
\|\jB\|_1 \leq d}} a_{\jB}(x) \cdot x^{\jB},
$$ where $a_{\jB}$ are periodic functions. If all the periodic coefficients $a_{\jB}$ are rational (or integer), then we call the whole quasi-polynomial $f$ \emph{rational} (resp. \emph{integer}). A period-vector $p$, which is common for all the coefficients $a_{\jB}(x)$, is called a \emph{period-vector} or a \emph{multi-period} of $f$.
%If $f$ is integer quasipolynomial, then the LCM of periods of all the coefficients $a_{j}(x)$ is called the \emph{quasiperiod of $f$}.
\end{definition}
Denote the projection from $\RR^{n_x + n_y}$ to the parametric space $\RR^{n_y}$ by $\Pi$. In other words, for $x \in \RR^{n_x}$ and $y \in \RR^{n_y}$, we have $\Pi\binom{x}{y} = y$.
\begin{definition}\label{chamber_decomp_def}
    Let $\PC$ be a polyhedron, defined by \ref{canonical_form} or \ref{standard_form}. Consider a collection $\QS$ of rational polyhedra $\QC$ with the following properties:
    \begin{enumerate}
        \item The equality $\Pi(\PC) = \bigcup\limits_{\QC \in \QS} \relint(\QC)$ is true;

        \item The equality $\relint(\QC_1) \cap \relint(\QC_2) = \emptyset$ is true, for different $\QC_1,\QC_2 \in \QS$;

        \item For any $\QC \in \QS$, all the polyhedra of the family $\{ \PC_{y} \colon y \in \QC \}$ have the same fixed combinatorial type;

        \item For any $\QC \in \QS$, there exists a subset of bases $\base(\QC)$ of $A$, such that, for any $y \in \relint(\QC)$, we have
        $$ 
        \vertex(\PC_{y}) = \bigl\{ A_{\BC}^{-1} (b_{\BC} - B_{\BC} y) \colon \BC \in \base(\QC)\bigr\}.
        $$ The set of corresponding affine functions
        $$
        \pvertex(\QC) := \bigl\{\VC_{\BC}(y) = A_{\BC}^{-1} (b_{\BC} - B_{\BC} y) \colon \BC \in \BS_{\QC}\bigr\}
        $$
        is called the \emph{parametric vertices of $\PC_y$.}
    \end{enumerate}
    The family $\QS$ is called the \emph{chamber decomposition of $\PC$}, elements of the decomposition are called \emph{chambers}. 
\end{definition}
Here, the symbol $\relint(\cdot)$ denotes the relative interior of the corresponding subset of $\RR^n$. More precisely, for $\PC \subseteq \RR^n$,
$$
\relint(\PC) = \{x \in \RR^n \colon \exists \varepsilon > 0 \text{ such that } x + \varepsilon \cdot \BC_{\norm{\cdot}_2} \cap \affh(\PC) \subseteq \PC \},
$$ where $\BC_{\norm{\cdot}_2}$ is the unit Euclidean ball and 
$
\affh(\PC)
$ is the affine hull of $\PC$.

\begin{definition}\label{denom_def}
    For a rational polyhedron $\PC \subseteq \RR^n$, the minimal value $q \in \ZZ_{\geq 1}$ (or $q \in \QQ_{>0}$), such that 
    $$
    \vertex(q \cdot \PC) \subseteq \ZZ^n
    $$ is called the \emph{integer denominator of $\PC$} (or resp. \emph{rational denominator of $\PC$}) and is denoted by $\denom_{\ZZ}(\PC)$ (resp. $\denom_{\QQ}(\PC)$). Clearly, $\denom_{\QQ}(\PC) \leq \denom_{\ZZ}(\PC)$.

    Let $\PC$ be a polyhedron, defined by \ref{canonical_form} or \ref{standard_form}, and let $\QS$ be a chamber decomposition of $\PC$. For a chamber $\QC \in \QS$, the minimal value $q \in \ZZ_{>0}$ (resp. $q \in \QQ_{>0}$), such that 
    $$
    q \cdot \VC(y) \in \ZZ^{n_x},\quad\text{for all $y \in \ZZ^{n_y}$ and $\VC \in \pvertex(\QC)$,}
    $$ is called the \emph{integer $($resp. rational$)$ chamber denominator of $\QC$} and is denoted by $\chdenom_{\ZZ}(\QC)$ (resp. $\chdenom_{\QQ}(\QC)$). Clearly, $\chdenom_{\QQ}(\QC) \leq \chdenom_{\ZZ}(\QC)$, for any $\QC \in \QS$.
\end{definition}

\begin{theorem}[Theorem~2, Clauss \& Loechner \cite{Parametric_Clauss}]\label{Z_multi_Ehrh_th}
    Let $\PC$ be a polyhedron, defined by \ref{canonical_form} or \ref{standard_form}. Then, the function $\ZEnum$ can be represented by an \emph{integer piece-wise quasi-polynomial} of degree $n_x$. 
    
    That is, there exists a chamber decomposition $\QS$, such that, for any $\QC \in \QS$ and $y \in \relint(\QC) \cap \ZZ^{n_y}$, the function $\ZEnum$ is an integer quasi-polynomial of degree $n_x$. The vector $\chdenom_{\ZZ}(\QC) \cdot \BUnit$ can be chosen as its period-vector.
\end{theorem}

\begin{definition}[Integer Piece-wise Ehrhart's Quasi-polynomial]\label{integer_Ehr_poly_def}
The representation of $\ZEnum$, given by the previous theorem, is called the \emph{integer piece-wise Ehrhart's quasi-polynomial of $\PC$}. The coefficients $a_{\jB}(y)$ can be interpreted as the integer periodic piece-wise defined functions. Since the number of chambers in the chamber decomposition $\QS$ is always finite, any of the coefficients $a_{\jB}(y)$ takes only a finite number of values. Therefore, we can directly store the values of all periodic coefficients $a_{\jB}(y)$. We call such a representation of $\ZEnum$ as the \emph{complete representation of the integer piece-wise Ehrhart's quasi-polynomial of $\PC$.} 
\end{definition}

Clauss \& Loechner \cite{Parametric_Clauss} give an algorithm to compute the part of chamber decomposition, i.e., only the chambers of dimension $n_y$. For the most applications, it is enough to know only the full-dimensional part of chambers. However, to establish our main results, we need to know the chambers of all dimensions. In our paper, we will propose an algorithm to compute them.

The following theorem, due to Henk \& Linke \cite{LinkeRational2} (see also Linke \cite{LinkeRational}), gives a generalization of Theorem \ref{Z_multi_Ehrh_th} with respect to the function $\QEnum$. Note that in the homogeneous case, when $b = 0$, all the chambers from the chamber decomposition of $\PC$ become polyhedral cones. %Additionally, note that our definition of the following theorem slightly differs from the original paper  \cite{LinkeRational2}
\begin{theorem}[Henk \& Linke \cite{LinkeRational2}]\label{Q_multi_Ehrh_th}
    Let $\PC$ be a polyhedron, defined by \ref{canonical_form} with $b = 0$ and $B = I_{m}$. In other words, $\PC_{y} = \{x \in \RR^{n_x} \colon A x \leq y\}$, for $y \in \RR^{n_y}$. Let $\QS$ be a chamber decomposition of $\PC$.

    Fix a chamber $\QC \in \QS$, and let $\QC = \cone(h_1, h_2, \dots h_s)$. Then, for $y \in \relint(\QC)$, the function $\QEnum$ can be represented by a rational quasi-polynomial
    $$
    \QEnum(y) = \sum\limits_{\substack{\jB \in \ZZ_{\geq 0}^{n_y}\\
\|\jB\|_1 \leq n_x}} a_{\jB}(y) \cdot y^{\jB},
    $$ where $a_{\jB}(y) = a_{\jB}\bigl(y + \denom_{\QQ}(\PC_{h_i}) \cdot h_i\bigr)$, for each $i \in \intint s$. That is, the matrix $D \cdot H$ is a period-matrix of the quasi-polynomial, where $H = (h_1\;\dots\;h_s)$ and $D$ is diagonal with $D_{i i} = \denom_{\QQ}(\PC_{h_i})$. 
    Furthermore, $a_{\jB}$ is a piece-wise defined polynomial of degree $n_x - \|\jB\|_1$ in $y$ with 
    $$
    \frac{\partial}{\partial y_i} a_{\jB} = - (\jB_i + 1) \cdot a_{\jB + e_i}.
    $$
\end{theorem}
\begin{definition}[Rational Piece-wise Ehrhart's Quasi-polynomial]\label{rational_Ehr_poly_def}
The representation of $\QEnum$, given by this theorem, is called the \emph{rational piece-wise Ehrhart's quasi-polynomial of $\PC$}.
\end{definition}

\begin{remark}\label{Linke_rm}
    It would be interesting to establish a simpler periodic property by analogy with Theorem \ref{Z_multi_Ehrh_th}, due to Clauss \& Loechner, such that
    $$
    \boxed{
    \text{for each $\QC \in \QS$, the vector $\chdenom_{\QQ}(\QC) \cdot \BUnit$ is a period-vector of $\QEnum$.}
    }
    $$ To the best of our knowledge, at the current moment of time such a property is not known. But, it seems that it can be deduced from the \emph{step-polynomial representation of $\QEnum$}, due to Verdoolaege, Seghir, Beyls, Loechner \& Bruynooghe \cite{CountingInParametricPolyhedra} (see also \cite{CountingFunctionEncoding}, due to Verdoolaege \& Woods), even for the non-homogeneous case with $b \not= \BZero$ and $B \not= I_{m}$. Notwithstanding this, in the Subsection \ref{quasipoly_subs}, Theorem \ref{my_real_Ehr_th} gives an independent proof of a slightly refined version of Theorem \ref{Q_multi_Ehrh_th}, due to Henk \& Linke, to prove this property for the general non-homogeneous case.
\end{remark}

Theorem \ref{Q_multi_Ehrh_th}, due to Henk \& Linke, has many generalizations that work with more general functions (evaluations) than $\QEnum$. For example, the weighted Minkowski sums of rational polyhedra can be represented as quasi-polynomials on weights, see Henk \& Linke \cite{LinkeRational2} and Stapledon \cite{Stapledon_counting} for algorithmic implications. The major generalization of $\QEnum$ is given by the notion of the \emph{intermediate weighted sums on polyhedra}: 
$$
\SC^{\LC}(\PC_{y}, h) = \sum\limits_{x \in \ZZ^{n_x}/\LC} \quad \int\limits_{\PC_{y} \cap (x + \LC)} h(t)\, dt,
$$ where $y \in \QQ^{n_y}$, $h(x)$ is a polynomial function, and $\LC$ is a rational linear subspace of $\RR^{n_x}$. It turns out that the structure of $\SC^{\LC}(\PC_{y}, h)$ can also be expressed by quasi-polynomials. The algorithmic theory (in a fixed dimension) of intermediate weighted sums on polyhedra is developed in the sequence of works \cite{RealEhrh2,ThreeEhrhart,RealEhrh}, due to Baldoni, Berline, De~Loera, K{\"o}ppe, Vergne, and \cite{RealEhrh3}, due to Beck, Elia \& Rehberg. Similar to Barvinok \cite{Simplex_Barv}, Baldoni et al. \cite{Simplex_Baldoni} give a polynomial-time algorithm to compute the highest coefficients of the corresponding quasi-polynomials.

{\bf The piece-wise step-polynomial representation of $\QEnum$.} As it was already mentioned, even for a fixed $n_x$, the integer or real Ehrhart's quasi-polynomials can not be used as an effective data structure to calculate $\ZEnum(y)$ or $\QEnum$, for a given $y \in \ZZ^{n_y}$ or $y \in \QQ^{n_y}$. For this reason, Verdoolaege, Seghir, Beyls, Loechner \& Bruynooghe \cite{CountingInParametricPolyhedra}, present an alternative representation of $\QEnum$, called the \emph{piece-wise step-polynomial}, which is, for a fixed chamber $\QC$, is a polynomial expression in the lower integer parts of the parametric vertices $\pvertex(\QC)$ of $\QC \in \QS$. The algorithm of \cite{CountingInParametricPolyhedra} computes such a representation of $\QEnum$ by a polynomial-time algorithm, assuming that $n_x$ and $n_y$ are fixed. If $n_x$ is fixed, the length of the resulting representation is bounded by a polynomial on the input size, which gives a practically good query time to compute $\QEnum(y)$, for a given $y \in \QQ^{n_y}$. Due to Verdoolaege \& Woods \cite{CountingFunctionEncoding}, the class of piece-wise step-polynomials and the class of rational generating functions are equivalent in the following sense: both representations can be transformed to each other by a polynomial-time algorithm in the assumption that $n_x$ and $n_y$ are fixed.

{\bf The dual principle.} All the considered algorithms are called primal-type counting algorithms. The dual-type counting algorithms are originally applied to polytopes $\PC$, defined in the standard form $\PC = \{x \in \RR_{\geq 0}^n \colon A x = y \}$, where $A \in \ZZ^{k \times n}$, $\rank(A) = k$, and $y \in \ZZ^k$. To the best of our knowledge, the dual-type generating-function framework was initiated by Brion \& Vergne \cite{brion_simple_poly,brion_residue}, Beck \cite{beck_residue,beck_partial_fractions,beck_dedekind_sums}, Nesterov \cite{knapsack_nesterov}, and Lasserre \& Zeron \cite{AlternativeCounting,fixedM_counting_lasserre}; see the monograph \cite{counting_Lasserre_book} by Lasserre. Denote $f_{A}(y) = |\PC_{y} \cap \ZZ^{n}|$ and consider the \emph{Z-transform} $\hat f_{A}(\zB) := \sum\limits_{y \in \ZZ^k} f_A(y) \cdot \zB^y$. Brion \& Vergne \cite{brion_residue} showed that $\hat f_A$ admits a simple closed formula $\hat f_A(\zB) = \prod\limits_{i=1}^n \frac{1}{1-\zB^{-A_{* i}}}$ and that  $f_A(b)$ can be recovered by the \emph{inverse Z-transform}, which is a multi-dimensional contour integration of $\hat f_A$.  Using this technique, Lasserre \& Zeron \cite{AlternativeCounting} present an algorithm to find $f_A(y)$ with the complexity bound $O(k)^d \cdot \Lambda$, where $d := \dim(\PC(b)) = n - k$ and $\Lambda$, where the parameter $\Lambda$ depends as a polynomial on $m$, $d$, and $\Delta_1$, but exponentially on the input size. The last bound was improved in \cite{SparseILP_Gribanov} to the bound $O\bigl(\frac{k}{d}+1\bigr)^{d/2} \cdot d^3 \cdot \Delta^3$, which additionally can be used for general polytopes of the co-dimension $k$, defined by both standard and canonical forms.

\section{New representation: Piece-wise periodic step-polynomials}\label{new_repr_sec}

In our work, we introduce the class of \emph{piece-wise periodic step-polynomials}, which differs from standard piece-wise step-polynomials by periodicity of the coefficients. The period-vector of any coefficient in our representation has smaller components than in the rational piece-wise Ehrhart's polynomial representation. More precisely, the product of the multi-period components of any coefficient is bounded by $\Delta$. The total length of our new representation can be even polynomial on $n_x$ in some important cases. Following to the papers \cite{Parametric_Clauss,CountingFunctionEncoding}, let us make some definitions.
\begin{definition}
Given real vector spaces $\VC$ and $\WC$, the function $\TC \colon \VC \to \WC$ of the form
    $$
    \TC(x) = \bigl\lfloor \AC(x) \bigr\rfloor,
    $$ where $\AC \colon \VC \to \WC$ is an affine map, is called \emph{the affine step-function}. 
\end{definition}

\begin{definition}
    Given real vector spaces $\VC$ and $\WC$ and an integer lattice $\Lambda \subseteq \WC$, a \emph{periodic step-polynomial} $f \colon \VC \to \RR$ is a function of the form
    \begin{equation*}
        f(x) = \sum\limits_{i = 1}^l \pi_{i}\bigl(\TC_i(x)\bigr) \cdot \Bigl(\LC_i\bigl(\TC_i(x)\bigr)\Bigr)^{d_i},
    \end{equation*}
    where, for any $i \in \intint l$, $\TC_i \colon \VC \to \WC$ are affine step-functions, $\LC_i \colon \WC \to \RR$ are linear functions, $\pi_{i} \colon \Lambda \to \RR$ are periodic functions 
    %(e.g. $\pi_{i}(x) = \pi_{i}(x + t)$, for some $t \in \Lambda$) 
    and $d_i \in \ZZ_{\geq 0}$. We say that \emph{the degree} of $f(x)$ is $\max_i\{d_i\}$ and \emph{the length} of $f(x)$ is $l$. 
    %A vector $q$ is called a \emph{multi-period} of $f(x)$, if $q$ is a multi-period for all $\pi_i$.
\end{definition}

\begin{definition}\label{periodic_step_poly_def}
Let $\DC$ be a subset of a real vector space and $\QS$ be a family of rational polyhedra, such that their relative interiors form a partition of $\DC$. Then, a function $f \colon \DC \to \RR$ with the property:
    $$
    \text{for each $\QC \in \QS$, the function $\restr{f}{\relint(\QC)}$ is a periodic step-polynomial,}
    $$
    is called a \emph{piece-wise periodic step-polynomial defined on $\QS$}. 
\end{definition}

The next theorem is the main theorem of our work. It states that there exists a piece-wise periodic step-polynomial representation of $\QEnum$ of a very special structure. Additionally, it presents an algorithm to compute this representation.
\begin{theorem}\label{main_param_th}
Let $\PC$ be a polyhedron, defined by \ref{canonical_form}. Assume that $\PC_{y}$ is bounded, for at least one $y \in \Pi(\PC)$. Then, there exists a chamber decomposition $\QS$ of $\PC$ and a piece-wise periodic step-polynomial $f$, defined on $\QS$, such that, for any $y \in \Pi(\PC)$ with bounded $\PC_{y}$, it holds $\QEnum(y) = f(y)$.

Additionally, the following propositions hold:
\begin{enumerate}
    \item For a base $\BC$ of $A$, denote $A_{\BC}=P_{\BC} S_{\BC} Q_{\BC}$, where $S_{\BC}$ is the SNF of $A_{\BC}$ and $P_{\BC},Q_{\BC} \in \ZZ^{n_x \times n_x}$ are unimodular. Then, for a fixed chamber $\QC \in \QS$ and $y \in \relint(\QC)$, we have
    \begin{equation}\label{counting_func_repr_eq}
    f(y) = \sum\limits_{\BC \in \base(\QC)} \sum\limits_{k = 0}^{n_x} \pi_{\BC,k}\bigl(P_{\BC}\TC_{\BC}(y)\bigr) \cdot \bigl\langle c_{\BC}, \TC_{\BC}(y) \bigr\rangle^k,
    \end{equation}
    where $\TC_{\BC}(y) \colon \RR^{n_y} \to \RR^{n_x}$ are affine step-functions, $c_{\BC} \in \QQ^{n_x}$, and $\pi_{\BC,k} \colon \ZZ^n \to \QQ_{\geq 0}$ are periodic functions with a period-matrix $S_{\BC}$. 
    %the formula $\pi_{\BC,k}(x) = \pi_{\BC,k}(x + S_{\BC} \cdot \BUnit_{n})$. 
    More precisely, for $\BC \in \base(\QC)$, the vector $c_{\BC}$ and the step-function $\TC_{\BC}(y)$ are given by the formulas:
    \begin{enumerate}
        \item $c_{\BC} = A_{\BC}^{-\top} c$, where $c \in \ZZ^n$ is some fixed integer vector;
        \item $\TC_{\BC}(y) = \bigl\lfloor A_{\BC} \VC_{\BC}(y) \bigr\rfloor$, where $\VC_{\BC}(y) = A_{\BC}^{-1} (b_{\BC} - B_{\BC} y)$ is the parametric vertex of $\PC$, corresponding to the base $\BC$. Consequently, $\TC_{\BC}(y) = \bigl\lfloor b_{\BC} - B_{\BC} y \bigr\rfloor$.
    \end{enumerate}

    \item Assume that $n_y$-dimensional faces and $(n_y-1)$-dimensional faces of $\PC$ are given.
%, the $0$-dimensional faces (vertices). 
Here, we assume that each face $\FC$ of $\PC$ is uniquely determined by a set of inequalities, which become equalities on $\FC$. Then, the function $f$ can be computed with the arithmetic cost 
$$
O^*\bigl((f_{n_y-1})^{n_y} \cdot f_{n_y} + (f_{n_y-1})^{2n_y} \cdot (f_{n_y-1} + \mu^2 \cdot \Delta^3)\bigr).
$$

\item The length and degree of the resulting piece-wise periodic step-polynomial is bounded by $\mu \cdot (n_x+1)$ and $n_x$, respectively.  Having such a representation, queries to $f$ can be performed with the cost of
$$
O\Bigl( n_y \cdot f_{n_y-1} +  \mu \cdot n_x \cdot \bigl(\log(\Delta) + n_y\bigr) \Bigr) \quad \text{operations.}
$$
\end{enumerate}
\end{theorem}
A theorem's proof is given in Section \ref{main_param_th_proof}. Its important part is a computation of the corresponding chamber decomposition of $\PC$. As it was already noted, the previous works only give algorithms to compute full-dimensional chambers, because it is sufficient for all the applications so far. However, for our needs, we need the full chamber decomposition of $\PC$, according to Definition \ref{chamber_decomp_def}. Since our algorithm is new, and, perhaps, it has an independent interest, we emphasize it to a separate theorem:
\begin{theorem}\label{chamber_decomp_th}
Let $\PC$ be a polyhedron, defined by \ref{standard_form} with $\rank(A) = n_x$ and $\dim(\PC) = n_x + n_y$. Assume that $n_y$-dimensional faces and $(n_y-1)$-dimensional faces of $\PC$ are given by lists of inequalities, which become equalities on a corresponding face. Then, the chamber decomposition of $\PC$ can be computed by an algorithm with arithmetic complexity bound:
$$
O^*\Bigl((f_{n_y-1})^{n_y} \cdot f_{n_y} + (f_{n_y-1})^{2n_y} \cdot (\nu + f_{n_y-1} )\Bigr).
$$ The total number of chambers is bounded by $O\bigl((f_{n_y-1})^{2 n_y}\bigr)$, the number of chambers of the dimension $n_y$ is bounded by $O\bigl((f_{n_y-1})^{n_y}\bigr)$. For a given point $y \in \Pi(\PC)$, the corresponding chamber $\QC$ with $y \in \relint(\QC)$ can be found with
$$
O(n_y \cdot f_{n_y-1})\quad\text{operations.}
$$ 
\end{theorem}
A theorem's proof is given in Section \ref{chamber_decomp_proof}. 

\subsection{Connection with the rational and integer Ehrhart's quasi-polynomials}\label{quasipoly_subs}

In this Subsection, we are going to show that the new piece-wise periodic step-polynomial and Ehrhart's piece-wise quasi-polynomial representations are closely connected. Moreover, the second one can be computed using the first. 

Consider the formula \eqref{counting_func_repr_eq} for a fixed chamber $\QC \in \QS$, and denote $\psi_{\BC}(y) = \bigl\langle c_{\BC}, b_{\BC} - \{ b_{\BC} - B_{\BC} y \} \bigr\rangle$. Clearly, $\bigl\langle c_{\BC}, \TC_{\BC}(y) \bigr\rangle = \psi_{\BC}(y) - \bigl\langle c_{\BC}, B_{\BC} y \bigr\rangle$. Substituting the last expression to \eqref{counting_func_repr_eq}, we have
\begin{multline}
    \QEnum(y) = \sum\limits_{\BC \in \base(\QC)} \sum\limits_{k = 0}^{n_x} \pi_{\BC,k}\bigl(P_{\BC}\TC_{\BC}(y)\bigr) \cdot \Bigl(\psi_{\BC}(y) - \bigl\langle c_{\BC}, B_{\BC} y \bigr\rangle\Bigr)^k = \\
    = \sum\limits_{\BC \in \base(\QC)} \sum\limits_{k = 0}^{n_x} \sum\limits_{i = 0}^k \pi_{\BC,k}\bigl(P_{\BC}\TC_{\BC}(y)\bigr) \binom{k}{i} \psi^{k-i}_{\BC}(y) \bigl\langle c_{\BC}, -B_{\BC} y \bigr\rangle^i = \\
    = \sum\limits_{i = 0}^{n_x} \sum\limits_{\BC \in \base(\QC)} \bigl\langle c_{\BC}, -B_{\BC} y \bigr\rangle^i \cdot \sum\limits_{k = i}^{n_x} \pi_{\BC,k}\bigl(P_{\BC}\TC_{\BC}(y)\bigr) \binom{k}{i} \psi^{k-i}_{\BC}(y) = \\
    = \sum\limits_{i = 0}^{n_x} \sum\limits_{\BC \in \base(\QC)} \bar \pi_{\BC,i}(y) \cdot \bigl\langle c_{\BC}, -B_{\BC} y \bigr\rangle^i,\label{period_quasipoly_coeff_eq} 
\end{multline}
where $\bar \pi_{\BC,i}(y) = \sum\limits_{k = i}^{n_x} \pi_{\BC,k}\bigl(P_{\BC}\TC_{\BC}(y)\bigr) \binom{k}{i} \psi^{k-i}_{\BC}(y)$.

Clearly, the expression $\bigl\langle c_{\BC}, - B_{\BC} y \bigr\rangle^i$ forms homogeneous polynomials on $y$. The next lemma shows that the coefficients $\bar \pi_{\BC,i}(y)$ are periodic functions. Consequently, the formula \eqref{period_quasipoly_coeff_eq} forms a quasi-polynomial.

\begin{lemma}\label{period_quasipoly_coeff_lm}
The following propositions hold:
    \begin{enumerate}
        \item The vectors $\chdenom_{\ZZ}(\QC) \cdot \BUnit$ and $\chdenom_{\QQ}(\QC) \cdot \BUnit$ could be chosen as period-vectors of $\bar \pi_{\BC,i}$;

        \item Let $y \in \QC$ and $z \in \relint(\QC)$. Denote $q = \lcm\bigl(\denom_{\QQ}(\PC_{y}), \denom_{\QQ}(\PC_{z})\bigr)$. Then, $\bar \pi_{\BC,i}(y) = \bar \pi_{\BC,i}\bigl(y + q \cdot (z-y)\bigr)$. 
    \end{enumerate}
\end{lemma}
\begin{proof}
    Denote $q = \chdenom_{\QQ}(\QC)$. Let us prove the first proposition. More precisely, we claim that any of the vectors $q \cdot t$, for $t \in \ZZ^{n_y}$, can be chosen as a period-vector of both functions. By definition of $\chdenom_{\QQ}(\QC)$, for any $\BC \in \base(\QC)$, we have 
    $$
    q \cdot A_{\BC}^{-1}(b_{\BC} - B_{\BC}  t) \in \ZZ^{n_x},\quad \text{for any $t \in \ZZ^{n_y}$}.
    $$ The last fact is possible if and only if $q \cdot A_{\BC}^{-1} b_{\BC} \in \ZZ^{n_x}$ and $q \cdot A_{\BC}^{-1} B_{\BC} t \in \ZZ^{n_x}$. Note, additionally, that $q \cdot B_{\BC} t \in \ZZ^{n_x}$. Therefore,
    \begin{multline*}
        \pi_{\BC,k}\bigl(P_{\BC}\TC_{\BC}(y + q t)\bigr) = \pi_{\BC,k}\bigl(P_{\BC} \lfloor b_{\BC} - B_{\BC} y \rfloor - q P_{\BC} B_{\BC} t \bigr) =\\ 
        = \pi_{\BC,k}\bigl(P_{\BC} \TC_{\BC}(y) - q P_{\BC} A_{\BC} A_{\BC}^{-1} B_{\BC} t \bigr) = \pi_{\BC,k}\bigl(P_{\BC} \TC_{\BC}(y) - S_{\BC} t' \bigr),
    \end{multline*}
    where $t' = q Q_{\BC}^{-1} A_{\BC}^{-1} B_{\BC} t$ is an integer vector, because $q \cdot A_{\BC}^{-1} B_{\BC} t \in \ZZ^{n_x}$ and $Q_{\BC}$ is unimodular. Since $t' \cdot S_{\BC}$ is a multi-period of $\pi_{\BC,k}$, the equality $\pi_{\BC,k}\bigl(P_{\BC}\TC_{\BC}(y + q \cdot t)\bigr) = \pi_{\BC,k}\bigl(P_{\BC}\TC_{\BC}(y)\bigr)$ holds. To finish the proof of the first proposition, we have left to show that $\psi_{\BC}(y + q \cdot t) = \psi_{\BC}(y)$. Definitely, due to the definition of $\psi_{\BC}(y)$, we just need to establish the equality $\{b_{\BC} - B_{\BC} (y + q \cdot t)\} = \{b_{\BC} - B_{\BC} y \}$, which holds since $q \cdot B_{\BC} t \in \ZZ^{n_x}$.

    Let us prove the second proposition. Since $y \in \QC$ and $z \in \relint(\QC)$, the polyhedra $\PC_y$ and $\PC_{z}$ have the same set of parametric vertices. Denote $t = z - y$. By definition of $q$, we have
    \begin{gather*}
    q \cdot A_{\BC}^{-1} \bigl(b_{\BC} - B_{\BC}y \bigr) \in \ZZ^{n_x},\\
    q \cdot A_{\BC}^{-1} \bigl(b_{\BC} - B_{\BC}z\bigr) \in \ZZ^{n_x}, \quad \text{for any $\BC \in \base(\QC)$}.
    \end{gather*}
    Consequently,
    \begin{gather*}
    q \cdot A_{\BC}^{-1} B_{\BC} t \in \ZZ^{n_x},\\
    q \cdot B_{\BC} t \in \ZZ^{n_x}, \quad \text{for any $\BC \in \base(\QC)$}.
    \end{gather*}
    Therefore, the same chain of reasoning can be used in a proof of the first proposition: 
    \begin{multline*}
        \pi_{\BC,k}\bigl(P_{\BC}\TC_{\BC}(y + q t)\bigr) = \pi_{\BC,k}\bigl(P_{\BC} \lfloor b_{\BC} - B_{\BC} y \rfloor - q P_{\BC} B_{\BC} t \bigr) =\\ 
        = \pi_{\BC,k}\bigl(P_{\BC} \TC_{\BC}(y) - q P_{\BC} A_{\BC} A_{\BC}^{-1} B_{\BC} t \bigr) = \pi_{\BC,k}\bigl(P_{\BC} \TC_{\BC}(y) - S_{\BC} t' \bigr) = \\
        = \pi_{\BC,k}\bigl(P_{\BC}\TC_{\BC}(y)\bigr).
    \end{multline*}
    A proof of the equality $\psi_{\BC}(y + q \cdot t) = \psi_{\BC}(y)$ is also completely similar.
\end{proof}

Therefore, we have proven the following theorem, which gives a direct generalization of Theorem \ref{Q_multi_Ehrh_th}, due to Henk \& Linke, modulo that we cannot say anything about the derivatives of $a_{\jB}$, because of the discrete nature of $\bar \pi_{\BC, k}$.
\begin{theorem}\label{my_real_Ehr_th}
    Let $\PC$ be a polyhedron, defined by \ref{canonical_form}. Then, there exists a chamber decomposition $\QS$ of $\PC$, such that, for any fixed $\QC \in \QS$ that corresponds to a bounded $\PC_{y}$ and $y \in \relint(\QC)$, the function $\QEnum$ can be  represented by a quasi-polynomial of degree $n_x$:
    $$
    \QEnum(y) = \sum\limits_{\substack{\jB \in \ZZ_{\geq 0}^{n_y}\\
\|\jB\|_1 \leq n_x}} a_{\jB}(y) \cdot y^{\jB}.
    $$ Additionally, the following propositions hold:
    \begin{enumerate}
        \item The vector $\chdenom_{\QQ}(\QC) \cdot \BUnit$ can be chosen as a period-vector of all $a_{\jB}$;

        \item For any $y \in \QC$ and $z \in \relint(\QC)$, if $q = \lcm\bigl(\denom_{\QQ}(\PC_{y}), \denom_{\QQ}(\PC_{z})\bigr)$, then $a_{\jB}(y) = a_{\jB}\bigl(y + q \cdot (z-y)\bigr)$ is true.
    \end{enumerate}
\end{theorem}

The formula \eqref{period_quasipoly_coeff_eq} can be used to establish the exact formula for $a_{\jB}(y)$, which can be used to estimate the computational complexity of $a_{\jB}(y)$. Definitely, due to \eqref{period_quasipoly_coeff_eq}:
\begin{multline*}
    \QEnum(y) = \sum\limits_{k = 0}^{n_x} \sum\limits_{\BC \in \base(\QC)} \bar \pi_{\BC,k}(y) \cdot \bigl\langle c_{\BC}, -B_{\BC} y \bigr\rangle^k =\\
    = \sum\limits_{k = 0}^{n_x} \sum\limits_{\substack{\jB \in \ZZ_{\geq 0}^{n_y} \\
    j_1 + \dots + j_{n_y} = k}} y^{\jB} \cdot \sum\limits_{\BC \in \base(\QC)} \binom{k}{j_1 \,\dots\, j_{n_y}} \cdot \bar \pi_{\BC,k}(y) \cdot (-c_{\BC}^{\top} B_{\BC})^{\jB}
\end{multline*}
\begin{multline}
    \text{Therefore,}\quad a_{\jB}(y) = \\
    = \sum\limits_{\BC \in \base(\QC)} \binom{k}{j_1 \,\dots\, j_{n_y}} \cdot \bar \pi_{\BC,k}(y) \cdot (-c_{\BC}^{\top} B_{\BC})^{\jB},\quad \text{where $k = j_1 + \dots + j_{n_y}$.} \label{aj_coeff_eq}
\end{multline}
In the next theorem, we estimate the computation complexity to evaluate the periodic coefficients $a_{\jB}(y)$.

\begin{theorem}\label{my_real_Ehr_comp_th}
    Assume that all the conditions of Theorem \ref{my_real_Ehr_th} are satisfied. Then, there exists a preprocessing algorithm with the arithmetic complexity 
    $$
    O^*\bigl((f_{n_y-1})^{n_y} \cdot f_{n_y} + (f_{n_y-1})^{2n_y} \cdot (f_{n_y-1} + \mu^2 \cdot \Delta^3)\bigr),
    $$ such that, for any $\QC \in \QS$, $\jB$, and $y \in \relint(\QC)$, the value $a_{\jB}(y)$ can be computed with
    $$
    O\Bigl(\mu \cdot n_x \cdot \bigl( \log(\Delta) + n_y\bigr)\Bigr) \quad\text{operations.}
    $$ The corresponding chamber $\QC$ can be found with $O\bigl(n_y \cdot f_{n_y - 1}\bigr)$ operations.
\end{theorem}
\begin{proof}
    To compute the value $a_{\jB}(y)$, we will use the formulas \eqref{period_quasipoly_coeff_eq} and \eqref{aj_coeff_eq}. First, we construct a piece-wise periodic polynomial representation of $\QEnum(y)$, using Theorem \ref{main_param_th}, with 
    \begin{equation}\label{my_real_Ehr_comp_eq1}
        O^*\bigl((f_{n_y-1})^{n_y} \cdot f_{n_y} + (f_{n_y-1})^{2n_y} \cdot (f_{n_y-1} + \mu^2 \cdot \Delta^3)\bigr).
    \end{equation}
     Additionally, we need to precompute the values $\binom{k}{j_1\,j_2\,\dots j_{n_y}}$, $\bigl(-c_{\BC}^{\top} B_{\BC}\bigr)^{\jB}$, where $\jB = \ZZ_{\geq 0}^{n_y}$, $j_1 + \dots + j_{n_y} = k$, $k \in \intint[0]{n_x}$, $\BC \in \base(\QC)$, and $\QC \in \QS$. It can be done with $O\bigl(N + |\QS| \cdot \mu \cdot  n_x \cdot n_y\bigr)$ operations, where $N$ is the total number of $\jB$-indices. Since, due to Theorem \ref{chamber_decomp_th}, $|\QS| = O\bigl((f_{n_y - 1})^{2 n_y}\bigr)$ and $N = O\bigl(n_x^{n_y}\bigr)$, the arithmetic cost of these additional computations is negligible with respect to \eqref{my_real_Ehr_comp_eq1}.

    For a given point $y \in \QQ^{n_y}$ and an index $\jB$, we first need to find a chamber $\QC \in \QS$, such that $y \in \relint(\QC)$. Due to Theorem \ref{chamber_decomp_th}, it can be done with $O(n_y \cdot f_{n_y - 1})$ operations. Next, for each $\BC \in \base(\QC)$, we use the following scheme:
    \begin{enumerate}
        \item Compute $\psi_{\BC}(y) = \bigl\langle c_{\BC}, b_{\BC} - \{b_{\BC} - B_{\BC} y\} \bigr\rangle$ with $O(n_x \cdot n_y)$ operations;

        \item For $i \in \intint[0]{n_x}$, compute $\bigl(\psi_{\BC}(y)\bigr)^i$ with $O(n_x)$ operations;

        \item Compute $g = P_{\BC} \TC_{\BC}(y) \bmod S_{\BC} \cdot \ZZ^{n_x}$ with $O(n_x \cdot \log(\Delta))$ operations; 
        % Compute $P_{\BC} \TC_{\BC}(y)$ with $O(n_x \cdot (n_x + n_y))$ operations;

        \item For $i \in \intint[0]{n_x}$, extract the values $\pi_{\BC, i}\bigl( g \bigr)$ with $O\bigl(n_x \cdot \log(\Delta)\bigr)$ operations;

        \item Compute $\bar \pi_{\BC,k}(y)$, where $k = j_1 + \dots + j_{n_y}$, with $O(n_x)$ operations, using the formula \eqref{period_quasipoly_coeff_eq}.
    \end{enumerate}
    The total arithmetic complexity of the presented scheme is $O\bigl(\mu \cdot n_x \cdot ( \log(\Delta) + n_y)\bigr)$. Now, we can compute the value $a_{\jB}(y)$, using the formula \eqref{aj_coeff_eq} with $O(\mu)$ operations.
\end{proof}

Due to Clauss \& Loechner's Theorem \ref{Z_multi_Ehrh_th}, the function $\ZEnum$ can be represented by an integer Ehrhart's piece-wise quasi-polynomial. Due to the periodicity reasons, its coefficient values can be stored exactly in a hash-table, which will give very fast evaluation time for the function $\ZEnum$. We can use the previous Theorem \ref{my_real_Ehr_comp_th} to compute this “complete” representation of $\ZEnum$, which is summarized in the following corollary:
\begin{corollary}\label{my_int_Ehr_comp_cor}
    Let $f(x)$ be an integer Ehrhart's piece-wise quasi-polynomial that represents $\ZEnum$, and let $\QS$ be the corresponding chamber decomposition. Denote $q = \sum_{\QC \in \QS} \bigl(\chdenom_{\ZZ}(\QC)\bigr)^{n_y}$. Assume additionally that the preprocessing step of Theorem \ref{my_real_Ehr_comp_th} has already been performed.
    
    Then, for all the chambers in $\QS$, we can precompute the values of all the periodic coefficients with
    $$
    O\Bigl(q \cdot M \cdot \mu \cdot n_x \cdot \bigl( \log(\Delta) + n_y\bigr)\Bigr)
    $$ operations, where $M = O(n_x^{n_y})$ is the maximum number of monomials. After that, for any $y \in \ZZ^{n_y}$, the value $f(y)$ can be computed with 
    $$
    O\bigl(n_y \cdot(f_{n_y -1} + M)\bigr)
    $$ operations.
\end{corollary}
\begin{proof}
    During the preprocessing, for each $\QC \in \QS$, we precompute values of all the corresponding coefficients $a_{\jB}$, using Theorem \ref{my_real_Ehr_comp_th}. We store these values in a hash-table with $O(1)$ lookup time and linear construction time. Since, for a fixed chamber, there are $M$ of such coefficients and since there are at most $\bigl(\chdenom_{\ZZ}(\QC)\bigr)^{n_y}$ unique values of a single coefficient, the total preprocessing cost is the same as it was claimed. 

    Now, for a given vector $y \in \ZZ^{n_y}$, we need first to find a chamber $\QC \in \QS$, such that $y \in \relint(\QC)$, due to Theorem \ref{chamber_decomp_th}, it can be done with $O(n_y \cdot f_{n_y -1})$ operations. After that, we look up values of the corresponding coefficients $a_{\jB}(y)$ and take a resulting sum with $O(n_y \cdot M)$ operations, which gives the desired complexity bound. 
\end{proof}

\section{Brief review of the obtained results}\label{summary_sec}
In the current Section, we review implications of our work from theoretical and computational perspectives.
{\newline \bf Theoretical perspective:}

\begin{enumerate}
    \item We show that the function $\QEnum$ can be represented by a new type of functions, named \emph{periodic piece-wise step-polynomials} (see Definition \ref{periodic_step_poly_def} and Theorem \ref{main_param_th}), which is a generalization of piece-wise step-polynomials from \cite{CountingInParametricPolyhedra,CountingFunctionEncoding}, due to Verdoolaege et al. Further, we show that the new representation of $\QEnum$ is more efficient than piece-wise step-polynomials and can be effectively computed in certain situations.
    
    \item We show that the rational piece-wise Ehrhart's quasi-polynomial representation of $\QEnum$ is a partial case of our new representation. More precisely, we give an independent proof of the main results of \cite{LinkeRational2} and \cite{LinkeRational}, due to Linke and Linke \& Henk, for the general non-homogeneous case. Additionally, we give some new information about multi-periods of the resulting piece-wise polynomials, based on the new notion of a \emph{chamber's denominator}. See Theorem \ref{Q_multi_Ehrh_th}, due to Henk \& Linke, our Theorem \ref{my_real_Ehr_th}, and Remark \ref{Linke_rm}.
\end{enumerate}
{\bf Computational perspective:}
\begin{enumerate}

\item {\bf General computational tool. } We give a general computational tool to construct a piece-wise periodic step-polynomial representation of $\QEnum$, which uses information on the face-lattice structure of $\PC$. This result is given in Theorem \ref{main_param_th}, and it is used to derive all the other consequences of our work. More precisely, assuming that faces of $\PC$ of dimensions $n_y$ and $n_y-1$ are given, the arithmetic complexity to construct the new representation of $\QEnum$ is
$$
O^*\bigl( (f_{n_y-1})^{n_y} \cdot f_{n_y} + (f_{n_y-1})^{2n_y} \cdot (f_{n_y-1} + \mu^2 \cdot \Delta^3) \bigr),
$$ while the arithmetic complexity to evaluate $\QEnum$, using our representation, in any given $y \in \QQ^{n_y}$, is
$$
O\bigl( n_y \cdot f_{n_y-1} + \mu \cdot n_x \cdot (\log\Delta + n_y)\bigr).
$$
% In the following table [???] we compare computational complexity of our approach with the approach of the papers \cite{CountingFunctionEncoding,CountingInParametricPolyhedra}. 

% \begin{center}
%     \begin{tabular}{|c|c|c|}
%     \hline
%     \hline
%          Preprocessing complexity & $\QEnum$-evaluation complexity &  \\
%          \hline
%          \hline
%     $O^*(p^{O(n_y)} \cdot \mu \cdot 2^{O(n_x)} \cdot (\log \Delta)^{n_x \ln n_x})$ & $O(\mu \cdot (\log \Delta)^{n_x \ln n_x})$ & due to \cite{CountingInParametricPolyhedra} \\
%     \hline
%     $O^*(p^{O(n_y)} \cdot \mu^2 \cdot \Delta^3)$ & $O(n_y \cdot p + \mu \cdot n_x \cdot (\log \Delta + n_y))$ & {\color{red} this work} \\
%     \hline
%     \hline
%     \end{tabular}
% \end{center}
Additionally, we give similar complexity bounds to compute coefficient values of the rational piece-wise Ehrhart's quasi-polynomial of $\PC$ and the complete integer piece-wise Ehrhart's quasi-polynomial representation of $\PC$. See Theorem \ref{my_real_Ehr_comp_th}, Corollary \ref{my_int_Ehr_comp_cor}, and Definitions \ref{integer_Ehr_poly_def}, \ref{rational_Ehr_poly_def}.

    \item {\bf Complexity bounds for polyhedra of bounded co-dimension.} Consider the polyhedron $\PC$, defined by a system
\begin{equation}\label{fixed_codim_system}
\begin{cases}
    A x = y\\
    x \in \RR_{\geq 0}^{n_x}\\
    y \in \RR^{k},
\end{cases}
\end{equation}
where $A \in \ZZ^{k \times n_x}$ and $\rank(A) = k$. This system can be considered as the “worst case” of the general parametric system in \ref{standard_form}, because it is natural to assume that $n_y \leq k$.

The systems of the type \eqref{fixed_codim_system} with a fixed $y \in \ZZ^{k}$ and a fixed co-dimension $k$ have received a considerable amount of attention in the literature. In his seminal work \cite{Papadimitriou}, Papadimitriou shows that ILP problems with systems \eqref{fixed_codim_system} can be solved by $\poly(n_x,\|A\|_{\max},\|y\|_{\infty})$-time algorithm, for any fixed $k$. The result of Papadimitriou was significantly refined by Jansen \& Rohwedder \cite{OnIPAndConv}, where $\poly\bigl(n_x,\|A\|_{\max},\log(\|y\|_{\infty})\bigr)$-time algorithm with a significantly better asymptotic behavior on $n_x$, $k$, and $\|A\|_{\max}$ was presented. Following to Jansen \& Rohwedder \cite{OnIPAndConv} or Eisenbrand \& Weismantel \cite{SteinitzILP}, in order to solve an ILP problem, one can reformulate the original system in such a way that the r.h.s. vector $y$ of the new system will depend only on $k$ and $\|A\|_{\max}$. Consequently, ILP problems with systems of the type \eqref{fixed_codim_system} can be solved in $\poly(n_x,\|A\|_{\max})$-time, for any fixed $k$. Moreover, the final complexity bound is FPT with respect to $k$ and $\|A\|_{\max}$.

Considering the counting problem, for any fixed $k$, Lasserre \& Zeron \cite{fixedM_counting_lasserre} present a dual-type algorithm that uses $\poly(n_x,\|A\|_{\max})$ arithmetic operations with real numbers. Unfortunately, the complexity analysis of this algorithm is not completely finished. The first $\poly(n_x,\|A\|_{\max})$-time algorithm with a complete complexity analysis was presented in \cite{Counting_FPT_Delta}, see also \cite{Counting_FPT_Delta_corrected}, since the original paper contained an inaccuracy.

As one of the main results of the current paper, we show that after \\$\poly(n_x,\|A\|_{\max})$-operations of a preprocessing algorithm, the parametric counting can be performed by a polynomial-time algorithm:
\begin{proposition}\label{fixed_k_main_result}
    Let $\PC$ be a polyhedron, defined by the system \eqref{fixed_codim_system}. Assume that the co-dimension $k$ is fixed. Then, there exists an $\poly\bigl(n, \|A\|_{\infty}\bigr)$-operations algorithm that returns a function $f : \RR^k \to \ZZ_{\geq 0} \cup \{+\infty\}$, such that $\QEnum(y) = f(y)$, for any $y \in \RR^k$. For any $y \in \QQ^k$, the value of $f(y)$ can be computed by a polynomial-time algorithm.
\end{proposition}

Similar results can be formulated with respect to the rational piece-wise Ehrhart's quasi-polynomial of $\PC$. 
\begin{proposition}\label{fixed_k_main_result2}
    Let $\PC$ be a polyhedron, defined by the system \eqref{fixed_codim_system}. Let $f(y)$ be the corresponding rational piece-wise Ehrhart's quasi-polynomial representation of $\QEnum$ with its chamber decomposition $\QS$. Assume that the co-dimension $k$ is fixed, then there exists an $\poly\bigl(n, \|A\|_{\infty}\bigr)$-operations preprocessing algorithm, which allows to compute any of the coefficients $a_{\jB}(y)$ by a polynomial-time algorithm, for any given chamber $\QC \in \QS$ and $y \in \relint(\QC)$.
\end{proposition}
For the problem to compute the complete integer piece-wise Ehrhart's quasi-polynomial representation of $\ZEnum$, we use a weaker parameter $\Delta_{\lcm} := \Delta_{\lcm}(A)$.
\begin{proposition}\label{fixed_k_main_result3}
    Let $\PC$ be a polyhedron, defined by the system \eqref{fixed_codim_system}, and $f(y)$ be the corresponding integer piece-wise Ehrhart's quasi-polynomial representation of $\ZEnum$. Assume that the co-dimension $k$ is fixed, then the complete representation of $f(y)$ can be computed by an $\poly(\Delta_{\lcm}, n_x)$-operations algorithm.
\end{proposition}
All the presented propositions are the straight corollaries of the general Theorems \ref{fixed_k_main_th}, \ref{fixed_k_main_th2}, and \ref{fixed_k_main_th3}, which are valid for any systems of the types \ref{canonical_form} and \ref{standard_form}.

\item {\bf General $\Delta$-modular polyhedra of a small co-dimension $k$.} As it was noted before, the results of the previous item are the partial cases of the more general Theorems \ref{fixed_k_main_th}, \ref{fixed_k_main_th2}, and \ref{fixed_k_main_th3}. More precisely, due to Theorem \ref{fixed_k_main_th}, for any parametric polyhedron $\PC$, defined by a system in the \ref{canonical_form} or \ref{standard_form} form, a piece-wise periodic step-polynomial representation of $\QEnum$ can be constructed with
$$
O(n_x/k)^{2k(n_y + 1)} \cdot \Delta^3 \cdot \poly(n_x,n_y,k)
$$ operations. After that, for any $y \in \QQ^{n_y}$, the value of $\QEnum(y)$ can be computed with
$$
O(n_x/k)^{k+1} \cdot n_x \cdot \bigl(\log(n_x \Delta) + n_y\bigr)
$$
operations. Similar results with respect to integer/rational piece-wise Ehrhart's quasi-polynomials are given in Theorems \ref{fixed_k_main_th2} and \ref{fixed_k_main_th3}.
 
\item {\bf Polyhedra of a general type with bounded dimension $n_y$ of the parametric space.}
Let us assume that the parameter $n_y$ is a fixed constant and consider the Parametric Counting problem for the general class of polyhedra, defined by systems in the \ref{canonical_form} or \ref{standard_form} forms. In the following Table \ref{complexity_ny_tb}, we compare our complexity bound (for the precise bound, see Theorem \ref{fixed_ny_main_th1}) with the approach from \cite{CountingInParametricPolyhedra} and \cite{CountingFunctionEncoding}, due to Verdoolaege et al.
% The justification of the complexity bound of the algorithm used in the paper \cite{CountingInParametricPolyhedra} is presented in Subsection \ref{VSBLB_just_rm}.
\begin{center}
    \begin{tabular}{|c|c|c|}
    \hline
    \hline
         Preprocessing complexity & $\QEnum$-evaluation complexity &  \\
         \hline
         \hline
    $m^{O(n_x)} \cdot \mu \cdot (\log \Delta)^{n_x \ln n_x}$ & $O^*(f_{n_y-1} + \mu \cdot (\log \Delta)^{n_x \ln n_x})$ & see \cite{CountingInParametricPolyhedra} \\
    \hline
    $m^{O(n_x)} \cdot \mu^2 \cdot \Delta^3$ & $O(f_{n_y - 1} + \mu \cdot n_x \cdot \log(n_x \Delta))$ & {\color{red} this work} \\
    \hline
    \hline
    \end{tabular}

    \captionof{table}{Complexity for a fixed $n_y$}\label{complexity_ny_tb}
\end{center}
Here, in Table \ref{complexity_ny_tb}, we use our Theorem \ref{chamber_decomp_th} to construct a data structure that stores and accesses the chambers from the chamber decomposition of $\PC$ for both approaches. Due to Theorem \ref{chamber_decomp_th}, assuming $n_y = O(1)$, the decomposition can be constructed with $n_x^{O(n_x)}$ operations, while the access costs is $O(f_{n_y-1})$. 
% The preprocessing complexity of the approach due to the paper \cite{CountingInParametricPolyhedra} is analyzed by the following way: for each chamber $\QC \in \QS$, we do a triangulation of the normal corresponding to $\PC_y$, for $y \in \relint(\QC)$ (which is the same for any $y \in \relint(A)$). The complexity of the last step can be estimated by $m^{O(n_x)} \cdot \mu$. Next, for each cone of the resulting triangulation, the Barvinok's algorithm is applied 

As we can see from Table \ref{complexity_ny_tb}, our approach has a better evaluation complexity, while the preprocessing step is competitive only for bounded values of $\Delta$. There are two interesting scenarios, when the evaluation complexity becomes $2^{O(n_x)}$:
\begin{enumerate}
    \item {\it Polyhedra with linear number of facets, i.e. $m = O(n_x)$.} In this situation, we clearly have $\mu,f_{n_y -1} = 2^{O(n_x)}$, so the bounds of Table \ref{complexity_ny_tb} become
    \begin{center}
    \begin{tabular}{|c|c|c|}
    \hline
    \hline
         preprocessing complexity & $\QEnum$-evaluation complexity &  \\
         \hline
         \hline
    $n_x^{O(n_x)} \cdot (\log \Delta)^{n_x \ln n_x}$ & $2^{O(n_x)} \cdot (\log \Delta)^{n_x \ln n_x}$ & see \cite{CountingInParametricPolyhedra} \\
    \hline
    $n_x^{O(n_x)} \cdot \Delta^3$ & $2^{O(n_x)}$ & {\color{red} this work} \\
    \hline
    \hline
    \end{tabular}

    \captionof{table}{complexity for $m=O(n_x)$}\label{complexity_mnx_tb}
\end{center}

    \item {\it Elements of $(A\, B)\footnote{We assume here that $B$ is integer.}$ are bounded, and $(A\, B)$ is sparse.} Denote $M = (A\,B)$. Let $r$ and $l$ denote the maximum number of non-zero elements in rows and columns of non-degenerate square sub-matrices of $M$, respectively. It was shown in \cite{SparseILP_Gribanov} that any triangulation of a cone, induced by a sparse matrix, has a bounded size, see the $3$-rd proposition of Lemma \ref{nu_mu_lm} of this work. Therefore, the following inequalities hold: $f_0 \leq \norm{M}_{\max}^{n_x + n_y} \cdot \min\{r,l\}^{n_x + n_y}$ and $\mu \leq \norm{A}_{\max}^{n_x} \cdot \min\{r,l\}^{n_x}$. Moreover, due to the Hadamard's bound, $\Delta \leq \norm{A}_{\max}^{n_x} \cdot \min\{r,l\}^{n_x/2}$. Therefore, assuming that $\norm{M}_{\max} = O(1)$, $f_{n_y-1}$ can be bounded by $\binom{f_0}{n_y - 1} = \min\{r,l\}^{O(n_x)}$, and the bounds of Table \ref{complexity_ny_tb} become
    \begin{center}
    \begin{tabular}{|c|c|c|}
    \hline
    \hline
         preprocessing complexity & $\QEnum$-evaluation complexity &  \\
         \hline
         \hline
    $m^{O(n_x)}$ & $n_x^{O(n_x)}$ & see \cite{CountingInParametricPolyhedra} \\
    \hline
    $m^{O(n_x)}$ & $\min\{r,l\}^{O(n_x)}$ & {\color{red} this work} \\
    \hline
    \hline
    \end{tabular}

    \captionof{table}{The complexity for sparse $A$ with bounded elements}\label{complexity_sp_tb}
\end{center}

For example, if $B$ is an identity matrix and $A$ is an incidence matrix of some hypergraph with a fixed maximum vertex degree or with a fixed maximum edge cardinality, or just an incidence matrix of some simple graph, then the evaluation complexity can be bounded by $2^{O(n_x)}$.
\end{enumerate}

\item {\bf Possible applications for the compiler analysis.}
Finally, let us make some speculative look of possible applications of our results to the compiler analysis. The classical work \cite{Parametric_Clauss}, due to Clauss \& Loechner, gives several examples, which illustrate how the parametric counting problem can be used for the compiler analysis. Let us consider only one of such examples, which is about the problem to estimate the nested loop execution time. Following to \cite{Parametric_Clauss}, consider the following loop nest:
\medskip
\begin{verbatim}
for i := 0 to n do
    for j := 0 to 1+i + m/2 do
        for k := to i - n + p - 1 do
            Statement
\end{verbatim}

We want to compute the number of flops in order to evaluate the execution time of this code segment. The loop nest is modeled by the parametric polytope $\PC = \{ x \in \ZZ^3 \colon A x \leq B y + b\}$, where 
% $x^\top = (i, j, k)$, $y^\top = (n, m, p)$, 
\begin{gather*}
    x^\top = (i, j, k),\quad y^\top = (n, m, p),\\
    A = \begin{pmatrix}
        -1 & 0 & 0 \\
        1 & 0 & 0 \\
        0 & -1 & 0 \\
        -2 & 2 & 0 \\
        0 & 0 & -1 \\
        -1 & 0 & 1
    \end{pmatrix}, \quad B = \begin{pmatrix}
        0 & 0 & 0 \\
        1 & 0 & 0 \\
        0 & 0 & 0 \\
        0 & 1 & 0 \\
        0 & 0 & 0 \\
        -1 & 0 & 1
    \end{pmatrix}, \quad b = \begin{pmatrix}
        0 \\ 0 \\ 0 \\ 2\\ 0 \\ -1
    \end{pmatrix}.
\end{gather*}
Therefore, for any given $y \in \ZZ^3$, the value of $\ZEnum(y)$ gives the exact number of the statement's evaluations.   

Now, let us make some, for our opinion, natural assumptions on a programmer's code:
\begin{enumerate}
    \item The nested loop defines a constant number of $i, j, k, \dots$-variables per  a single nested level. Definitely, in the most situations, a programmer uses only one variable (or a constant number of variables) per a single nested level. In terms of the parametric polyhedron $\PC$, it means that $m = O(n_x)$;  

    \item The coefficients of the $i, j, k, \dots$-variables in the nested loop are bounded by some fixed constant. Here, we assume that in a “common” case, the coefficients of the variables are “small”. With respect to $\PC$, it implies that $\Delta = n_x^{O(n_x)}$.
\end{enumerate}

Assuming that $n_y$, i.e., the number of $n,m,p,\dots$-variables, is fixed, let us compare the complexity of our approach with respect to the approach of Verdoolaege et al. \cite{CountingInParametricPolyhedra}. Due to Table \ref{complexity_mnx_tb}, we have the following complexity in the considered case:
\begin{center}
    \begin{tabular}{|c|c|c|}
    \hline
    \hline
         Preprocessing complexity & $\QEnum$-evaluation complexity &  \\
         \hline
         \hline
    $n_x^{O(n_x)}$ & $n_x^{O(n_x)}$ & see \cite{CountingInParametricPolyhedra} \\
    \hline
    $n_x^{O(n_x)}$ & $2^{O(n_x)}$ & {\color{red} this work} \\
    \hline
    \hline
    \end{tabular}
\end{center}
Therefore, in the considered “common” scenario, the theoretical $\QEnum$-evaluation complexity is better for our approach.

\end{enumerate}

\section{The computational complexity for special cases}\label{special_sec}

In this Section, we are going to apply Theorem \ref{main_param_th} to estimate the complexity of the parametric counting problem in two different scenarios. Additionally, we analyze the complexity to compute Ehrhart's quasi-polynomials for them. More precisely, we consider the following classes:
\begin{enumerate}
    \item Polyhedra, defined by systems of a bounded co-dimension;
    \item Polyhedra of the general type with a bounded dimension $n_y$ of the parametric space.
\end{enumerate}

\subsection{Polyhedra, defined by systems of a bounded co-dimension}\label{codim_subs}

In the following result, we use our main Theorem \ref{main_param_th} to construct a complexity bound for the parametric counting problem with respect to parametric polyhedra of a small co-dimension. The cases, when $\PC_y$ is unbounded, for some $y \in \QQ^{n_y}$, are handled, using Lemmas \ref{lines_eliminating_lm} and \ref{param_bounding_lm}.
\begin{theorem}\label{fixed_k_main_th}
    Let $\PC$ be a polyhedron, defined by a system in the \ref{standard_form} or \ref{canonical_form} of the co-dimension $k$. Then, there exists an algorithm with the arithmetic complexity bound
    $$
    O(n_x/k)^{2 k (n_y + 1)} \cdot \Delta^3 \cdot \poly(n_x,n_y,k),
    $$ which returns a piece-wise periodic step-polynomial $f(y) \colon \RR^{n_y} \to \ZZ_{\geq 0} \cup \{+\infty\}$, such that $\QEnum(y) = f(y)$. For any $y \in \QQ^{n_y}$, the value of $f(y)$ can be computed with 
    $$
    O(n_x/k)^{k+1} \cdot n_x \cdot\bigl(\log(n_x \Delta) + n_y\bigr)\quad\text{operations}.
    $$
\end{theorem}
\begin{proof}
    % Consider first the case, when $r := \rank(A) < k$ and assume that the first $r$ rows of $A$ are linearly independent. Using the Gaussian elimination, we transform the original system $Ax = b + By$ to an equivalent system $\binom{A_{1:r}}{\BZero} x = \binom{b_{1:r}}{\hat b} + \binom{B_{1:r}}{\hat B} y$, where the sub-system $A_{1:r} x = b_{1:r} + B_{1:r}$ represents the first $r$ lines of the original system. Clearly, the new system is feasible (even in $\RR^{n_x}$) only if $\hat b + \hat B y = \BZero$. Since the check of the equality $\hat b + \hat B y = \BZero$ can be performed by a polynomial-time algorithm for any $y \in \QQ^{n_y}$, we continue our work only with the reduced system $A_{1:r} x = b_{1:r} + B_{1:r} y$ of rank $k$.
    
    % Assume next that $\rank(A) = k$ and denote $d = n_x - k$, $m = d + k$. Due to \cite[Lemma~4 and Lemma~5]{OnCanonicalProblems_Grib}, the system in the \ref{standard_form} can be polynomially transformed to an equivalent system $A' x \leq b' + B' y$ in the \ref{canonical_form} with $A' \in \ZZ^{m \times d}$, $b' \in \QQ^{m}$ and $B' \in \QQ^{m \times n_y}$ with $\rank(A') = d$ and $\Delta(A') = \Delta(A)$. Let $\PC$ be the parametric polyhedron defined by the new system. Due to Remark \ref{B_rank_rm}, we can assume that $\rank(B') = n_y$. Due to Lemmas \ref{lines_eliminating_lm}, \ref{param_bounding_lm}, we can assume that $\PC_y$ is bounded, for any $y \in \Pi(\PC)$. The last property is achieved at the cost of replacing of $m = d + k$ by $m = d + k +1$, and $\Delta$ by $d \cdot \Delta$.

    Consider first the case, when $\PC$ is defined by a \ref{canonical_form}-system. Due to Remark \ref{B_rank_rm}, we can assume that $\rank(B) = n_y$. Due to Lemmas \ref{lines_eliminating_lm} and \ref{param_bounding_lm}, we can assume that $\PC_y$ is bounded, for any $y \in \Pi(\PC)$. The last property is achieved at the cost of replacing of $m = n_x + k$ by $m = n_x + k +1$ and $\Delta$ by $n_x \cdot \Delta$.

    We are going to use our main Theorem \ref{main_param_th}. To this end, we need to estimate the values $f_{n_y}$, $f_{n_y - 1}$, $\mu$ and the complexity to enumerate $n_y$-dimensional and $(n_y-1)$-dimensional faces of $\PC$. For our purposes, it is sufficient to use straightforward estimates for $f_{n_y}$, $f_{n_y - 1}$ and $\mu$. It follows that
    \begin{enumerate}
        \item $f_{n_y} \leq \binom{m}{(n_x + n_y)-n_y} = \binom{n_x + k +1}{k+1} = O(n_x/k)^{k+1}$;

        \item $f_{n_y-1} \leq \binom{m}{(n_x + n_y)-(n_y-1)} = \binom{n_x + k +1}{k} = O(n_x/k)^{k}$;

        \item $\mu \leq \binom{m}{n_x} = \binom{n_x+k+1}{k+1} = O(n_x/k)^{k+1}$.
    \end{enumerate}
    We enumerate all the $n_y$-dimensional and $(n_y-1)$-dimensional faces of $\PC$ just by straightforward enumeration of the corresponding sub-systems. Clearly, the complexity of such an enumeration procedure can be estimated by $O(n_x/k)^{k+1}\cdot \poly(n_x,n_y,k)$.

    Consider now the case, when $\PC$ is defined by a \ref{standard_form}-system. Assume that $r := \rank(A) < k$ and the first $r$ rows of $A$ are linearly independent. Using the Gaussian elimination, we transform the original system $Ax = b + By$ to an equivalent system $\binom{A_{1:r}}{\BZero} x = \binom{b_{1:r}}{\hat b} + \binom{B_{1:r}}{\hat B} y$, where the sub-system $A_{1:r} x = b_{1:r} + B_{1:r}$ represents the first $r$ lines of the original system. Clearly, the new system is feasible, even in $\RR^{n_x}$, only if $\hat b + \hat B y = \BZero$. Since, for any $y \in \QQ^{n_y}$, we can check the equality $\hat b + \hat B y = \BZero$ by a polynomial-time algorithm, we will continue our work only with the reduced system $A_{1:r} x = b_{1:r} + B_{1:r} y$ of rank $k$.
    
    % Since the check of the equality $\hat b + \hat B y = \BZero$ can be performed by a polynomial-time algorithm for any $y \in \QQ^{n_y}$, we continue our work only with the reduced system $A_{1:r} x = b_{1:r} + B_{1:r} y$ of rank $k$.

    Assume that $\rank(A) = k$ and denote $n_x' = n_x - k$, $m' = n_x' + k$. Due to Remark \ref{form_trans_rm}, any system in \ref{standard_form} can be polynomially transformed to an equivalent system $A' x \leq b' + B' y$ in the \ref{canonical_form} with $A' \in \ZZ^{m' \times n_x'}$, $b' \in \QQ^{m'}$ and $B' \in \QQ^{m' \times n_y}$, and with $\rank(A') = n_x'$ and $\Delta(A') = \Delta(A)$. To finish the proof, we just need to use the former reasoning to the new \ref{canonical_form}-system.
\end{proof}

In the next theorems, we present the same complexity analysis with respect to coefficients of the rational/integer piece-wise quasi-polynomials of $\PC$.
\begin{theorem}\label{fixed_k_main_th2}
    Let $\PC$ be a polyhedron, defined by a system in \ref{canonical_form} or \ref{standard_form} of the co-dimension $k$. Let $f(y)$ and $\QS$ be the corresponding rational piece-wise Ehrhart's quasi-polynomial representation of $\QEnum$ and its chamber decomposition. 

    Then, there exists a preprocessing algorithm with the arithmetic complexity bound 
    $$
    O(n_x/k)^{2 k (n_y + 1)} \cdot \Delta^3 \cdot \poly(n_x,n_y,k),
    $$ which allows to compute any of the coefficients $a_{\jB}(y)$ by an algorithm with the arithmetic complexity bound
    $$
    O(n_x/k)^{k+1} \cdot n_x \cdot\bigl(\log(n_x \Delta) + n_y\bigr),
    $$
    for any given chamber $\QC \in \QS$ and $y \in \relint(\QC)$. 
\end{theorem}

\begin{theorem}\label{fixed_k_main_th3}
    Let $\PC$ be a polyhedron, defined by a system in \ref{canonical_form} or \ref{standard_form} of the co-dimension $k$. Let $f(y)$ and $\QS$ be the corresponding integer piece-wise Ehrhart's quasi-polynomial representation of $\ZEnum$ and its chamber decomposition.

    Then, the complete representation of $f(y)$ can be computed by an algorithm with the complexity bound
    $$
    \Delta_{\lcm}^k \cdot M \cdot O(n_x/k)^{2 k n_y + n_y + k} \cdot \poly(\phi),
    $$ where $M = O(n_x^{n_y})$ is the maximum number of monomials and $\phi$ is the input size.
\end{theorem}
Proofs can be easily deduced from Theorem \ref{my_real_Ehr_comp_th} and Corollary \ref{my_int_Ehr_comp_cor} in the same way as Theorem \ref{fixed_k_main_th} has been deduced from Theorem \ref{main_param_th}.

\subsection{Polyhedra, defined by general-type systems}\label{nyfixed_subs}

In the following result, we use our main Theorem \ref{main_param_th} to construct a complexity bound for the parametric counting problem with respect to general parametric polyhedra, assuming that the dimension $n_y$ of the parametric space is fixed. The cases, when $\PC_y$ is unbounded, for some $y \in \QQ^{n_y}$, are again handled, using Lemmas \ref{lines_eliminating_lm} and \ref{param_bounding_lm}.

\begin{theorem}\label{fixed_ny_main_th1}
    Let $\PC$ be a polyhedron, defined by the \ref{canonical_form}. Then, the periodic piece-wise step-polynomial representation of the function $\QEnum$ can be computed by an algorithm with the complexity
    $$
    O^*\bigl( m^{\frac{n_x n_y}{2}} \cdot (m^{\frac{n_x}{2}} + \mu^2 \cdot \Delta^3) \bigr).
    $$ For any $y \in \QC^{n_y}$, the value of $\QEnum(y)$ can be found with
    $$
    O\bigl(n_y \cdot f_{n_y-1} + \mu \cdot n_x \cdot \bigl(\log (n_x \Delta) + n_y\bigr)\bigr)\quad\text{operations.}
    $$
\end{theorem}
\begin{proof}
    Due to Remark \ref{B_rank_rm} and Lemma \ref{lines_eliminating_lm}, we can assume that $\rank(B) = n_y$ and $\rank(A) = n_x$. Due to Lemma \ref{param_bounding_lm}, we can assume that $\PC_y$ is bounded, for any $y \in \Pi(\PC)$. The last property is achieved at the cost of replacing $m$ by $m+1$, and $\Delta$ by $n_x \cdot \Delta$, which can be hided by $O^*(\cdot)$ in the resulting complexity bound. Due to Remark \ref{dim_PC_rm}, we can assume that $d := \dim(\PC) = n_x + n_y$. 

    Let $\mu^* = \mu(\bigl( A\,B \bigr)^\top)$. Due to Lemma \ref{nu_mu_lm}, the value of $\mu^*$ can be estimated by $O\bigl(\frac{m}{d}\bigr)^{\frac{d+1}{2}} = O^*\bigl((\frac{m}{n_x})^{\frac{n_x}{2}}\bigr)$. Hence, due to D.~Avis \& K.~Fukuda \cite{AvisFukuda}, all the vertices of $\PC$ can be enumerated by an algorithm with the same complexity bound $O^*\bigl((\frac{m}{n_x})^{\frac{n_x}{2}}\bigr)$. Due to V.~Kaibel \& M.~Pfetsch \cite{FacesEnum_Kaibel}, all faces of dimension $\leq n_y$ can be enumerated by an algorithm with the complexity bound $O(m \cdot \alpha \cdot \phi^{\leq n_y}) = O^*\bigl((\frac{m}{n_x})^{\frac{n_x}{2}} \cdot \phi^{\leq n_y}\bigr)$, where $\phi^{\leq n_y}$ is the number of faces and $\alpha$ is the number of vertex-facet incidences. Due to \cite[Theorem~(7.4), p.~273]{Barv_book_convexity}, 
    \begin{equation*}
        f_k(\PC) = \sum\limits_{i = k}^{\lfloor d/2 \rfloor} \binom{i}{k} \binom{m-d+i-1}{i} + \sum\limits_{i = \lfloor d/2 \rfloor + 1}^{d} \binom{i}{k} \binom{m-i-1}{d-i}.
    \end{equation*}
    Since $n_y$ is a constant and $k \leq n_y$, it can be directly checked that $f_k(\PC) = O^*(m^{d/2}) = O^*(m^{n_x/2})$. Summarizing, we state that all the faces of $\PC$ of dimension $\leq n_y$ can be enumerated by an algorithm with the complexity bound $O^*(m^{n_x})$.

    % $O^*(m^{n_x+2})$, 

    Now, we can use the complexity bound of Theorem \ref{main_param_th}. Due to \ref{main_param_th}, the bound for $\phi^{\leq n_y}$, and the corresponding complexity bound, the periodic piece-wise step-polynomial representation of $\QEnum$ can be computed by an algorithm with the complexity bound
    \begin{equation*}
        O^*\bigl( m^{\frac{n_x n_y}{2}} \cdot (m^{\frac{n_x}{2}} + \mu^2 \cdot \Delta^3) \bigr).
    \end{equation*}
    % Again, due to Lemma \ref{nu_mu_lm}, $\mu(A) = O^*\bigl((\frac{m}{n_x})^{n_x/2}\bigr)$, the last complexity bound becomes
    % \begin{equation*}
    %     O^*\bigl( m^{\frac{n_x(n_y+2)}{2}} \cdot \Delta^3 \bigr).
    % \end{equation*}

    Finally, due to Theorem \ref{main_param_th}, for any $y \in \QQ^{n_y}$, the value of $\QEnum(y)$ can be evaluated by an algorithm with the complexity bound $O\bigl(n_y \cdot f_{n_y-1} + \mu \cdot n_x \cdot \bigl(\log(n_x \Delta) + n_y\bigr)\bigr)$, 
    % = O( \frac{m}{n_x} )^{n_x/2} \cdot n_x \cdot (\log \Delta + n_y)$, 
    which completes the proof.
\end{proof}

In the next theorem, we do the same complexity analysis with respect to coefficients of the rational piece-wise quasi-polynomial of $\PC$.
\begin{theorem}\label{fixed_ny_main_th2}
    Let $\PC$ be a polyhedron, defined by the \ref{canonical_form}. Let $f(y)$ and $\QS$ be the corresponding rational piece-wise Ehrhart's quasi-polynomial representation of $\QEnum$ and its chamber decomposition. Then, there exists a preprocessing algorithm with the complexity bound
    $$
    O^*\bigl( m^{\frac{n_x n_y}{2}} \cdot (m^{\frac{n_x}{2}} + \mu^2 \cdot \Delta^3) \bigr),
    $$ which allows to compute any of the coefficients $a_{\jB}(y)$ by an algorithm with the arithmetic complexity bound
    $$
    O\bigl(\mu \cdot n_x \cdot \bigl(\log (n_x \Delta) + n_y\bigr)\bigr),
    $$ for any given $\QC \in \QS$ and $y \in \relint(\QC)$. Given $y \in \QQ^{n_y}$, the corresponding chamber $\QC \in \QS$ with $y \in \relint(\QC)$ can be found with $O(n_y \cdot f_{n_y-1})$ operations.
\end{theorem}
A proof can be deduced from Theorem \ref{my_real_Ehr_comp_th} in the same way as Theorem \ref{fixed_ny_main_th1} has been deduced from Theorem \ref{main_param_th}.

\section{Preliminaries}\label{prelim_sec}

\subsection{Valuations and indicator functions of polyhedra}\label{valuations_subs}

In this Subsection, we mainly follow to the monographs \cite{BarvBook,BarvPom} in the most definitions and notations. Let $\VC$ be the Euclidean space and $\Lambda \subset \VC$ be an integer lattice.

\begin{definition}
Let $\AC \subseteq \VC$ be a set. The \emph{indicator} $[\AC]$ of $\AC$ is the function $[\AC]\colon \VC \to \RR$, defined by
$$
[\AC](x) = \begin{cases}
1\text{, if }x \in \AC\\
0\text{, if }x \notin \AC.
\end{cases}
$$ 
The \emph{algebra of polyhedra} $\PS(\VC)$ is the vector space, defined as the span of the indicator functions of all the polyhedra $\PC \subset \VC$.
\end{definition}

\begin{definition} Let $\PC \subseteq \VC$ be a set. The \emph{polar} $\PC^{\circ}$ of $\PC$ is defined by $$
\PC^\circ = \bigl\{ x \in \VC \colon \langle x,y \rangle \leq 1 \; \forall y \in \PC\bigr\},
$$ and the dual lattice is defined by
$$
\Lambda^{\circ} = \bigl\{ x \in \VC \colon \langle x, y \rangle \in \ZZ\; \forall y \in \Lambda \bigr\}.
$$
\end{definition}

\begin{definition}
The polyhedron $\PC \subseteq \VC$ is called \emph{rational} if it can be defined by a system of finitely many inequalities 
$$
\langle a_i, x \rangle \leq b_i,
$$
where $a_i \in \Lambda^{\circ}$ and $b_i \in \ZZ$. The \emph{algebra of rational polyhedra} $\PS(\QQ\VC)$ is the vector space, defined as the span of the indicator functions of all the rational polyhedra $\PC \subset \VC$.
\end{definition}

\begin{definition}
A linear transformation $\TC \colon \PS(\VC) \to \WC$, where $\WC$ is a vector space, is called a \emph{valuation}. We consider only \emph{$\LC$-valuations} or \emph{lattice valuations} that satisfy
$$
\TC\bigl([\PC + u]\bigr) = \TC\bigl([\PC]\bigr), \quad \text{for all rational polytopes }\PC \text{ and }u \in \LC,
$$ see \cite[pp. 933--988]{ValuationsAndDissections}, \cite{ValuationsOnConvecBodies}.
\end{definition}

\begin{remark}\label{quasipoly_interpolation} Let us denote $g(\PC) = \TC\bigl([\PC]\bigr)$, for a lattice valuation $\TC$. The general result of P.~McMullen \cite{LatticeInvariantValuations} states that if $\PC \subset \VC$ is a rational polytope, $d = \dim(\PC)$, and $t \in \NN$ is a number, such that $t \cdot \PC$ is a lattice polytope, then there exist functions $g_i(\PC,\cdot) \colon \ZZ_{\geq 0} \to \CC$, such that 
\begin{gather*}
g(\alpha \cdot \PC) = \sum\limits_{i=0}^d g_i(\PC, \alpha) \cdot \alpha^i, \quad\text{ for all } \alpha \in \ZZ_{\geq 0}\text{, and}\\
g_i(\PC, \alpha + t) = g_i(\PC, \alpha), \quad\text{ for all } \alpha \in \ZZ_{\geq 0}.
\end{gather*}
\end{remark}

\begin{theorem}[Theorem~2.3 of \cite{BarvPom}]\label{linear_map_eval_th}
Let $\VC$ and $\WC$ be finite-dimensional vector spaces, and let $T \colon \VC \to \WC$ be an affine transformation. Then 
\begin{enumerate}
\item[1)] For every polyhedron $\PC \subset \VC$, the image $T(\PC) \subset \WC$ is a polyhedron;
\item[2)] There is a unique linear transformation (valuation) $\TC \colon \PS(\VC) \to \PS(\WC)$, such that 
$$
\TC\bigl([\PC]\bigr) = \bigl[T(\PC)\bigr], \text{ for every polyhedron }\PC \subset \VC.
$$
\end{enumerate}
\end{theorem}
For $c \in \VC$, denote:
$$
\fG(\PC,c; \tau) = \sum\limits_{z \in \PC \cap \Lambda} e^{\langle c,z \rangle}.
$$
The first valuation $\FC\bigl([\PC]\bigr)$, which will be significantly used in our paper, is defined by the following restatement of the theorem, proved by J.~Lawrence \cite{Lawrence}, and, independently, by A.~Khovanskii and A.~Pukhlikov \cite{Pukhlikov}, declared as Theorem~13.8b in \cite[Section~13]{BarvBook}.
\begin{theorem}[Lawrence \cite{Lawrence}, Khovanskii \& Pukhlikov \cite{Pukhlikov}]\label{Pukhlikov_th}
    Let $\RS(\VC)$ be the space of functions in $\VC$, spanned by functions of the type
    $$
    \frac{e^{\langle c, v \rangle}}{\bigl(1 - e^{\langle c, u_1 \rangle}\bigr) \cdot \ldots \cdot \bigl(1- e^{\langle c, u_d \rangle}\bigr)},
    $$
    where $d = \dim(\VC)$, $v \in \Lambda$,  and $u_i \in \Lambda \setminus\{\BZero\}$, for $i \in \intint n$. Then, there exists a linear transformation (valuation)
    $$
    \FC \colon \PS(\QQ\VC) \to \RS(\VC),
    $$ such that the following properties hold:
    \begin{enumerate}
    
        \item[(1)] Let $\PC \subseteq \VC$ be a non-empty rational polyhedron without lines, and let $\RC := \RC_{\PC} \subseteq \VC$ be its recession cone. Then, for all $c \in \inter(K^{\circ})$, the series $\fG(\PC,c; \tau)$
        % $$
        % \sum\limits_{z \in \PC \cap \Lambda} e^{\langle c, z \rangle}
        % $$
        converges absolutely to a function $\FC\bigl([\PC]\bigr)$.
        
        \item[(2)] If $\PC$ contains a line, then $\FC\bigl([\PC]\bigr) = 0$.
        
    \end{enumerate}
\end{theorem}
If $\PC$ is a rational polyhedron, then $\FC\bigl([\PC]\bigr)$ is called its \emph{short rational generating function}.

\subsection{Vertices, edge directions, tangent cones, and triangulations}\label{vdtt_subs}

\begin{definition}
Let $\PC \subset \VC$ be a non-empty polyhedron, and let $v \in \PC$ be a point. The \emph{tangent cone} of $\PC$ at $v$ is defined as
$$
\tcone(\PC,v) = \bigl\{v+ y \colon v + \varepsilon y \in \PC, \; \text{ for some } \varepsilon > 0 \bigr\}.
$$
The \emph{cone of feasible directions} at $v$ is defined as
$$
\fcone(\PC,v) = \bigl\{y \colon v + \varepsilon y \in \PC, \; \text{ for some } \varepsilon > 0 \bigr\}.
$$
Thus, $\tcone(\PC,v) = v + \fcone(\PC,v)$.
\end{definition}

\begin{remark}\label{fcone_generator_rm}
If an $n$-dimensional polyhedron $\PC$ is defined by a system $A x \leq b$, then, for any $v \in \PC$, it holds
\begin{gather*}
\tcone(\PC,v) = \{ x \in \VC \colon A_{\JC(v) *} x  \leq b_{\JC(v)}\},\\
\fcone(\PC,v) = \{ x \in \VC \colon A_{\JC(v) *} x  \leq \BZero\},\\
\fcone(\PC,v)^\circ = \cone(A_{\JC(v) *}^\top),
\end{gather*}
where $\JC(v) = \{ j \colon A_{j *} v = b_j\}$.
\end{remark}

The famous M.~Brion's theorem \cite{Brion} connects the indicator function $[\PC]$ of a polyhedron with indicator functions of tangent cones, corresponding to its vertices. We take a formulation of Brion's theorem, presented in \cite[Theorem~6.4]{BarvBook}.
\begin{theorem}[M.~Brion \cite{Brion}]\label{Brion_th}
    Let $\PC \subseteq \VC$ be a polyhedron without lines. Then,
    $$
    [\PC] \equiv \sum\limits_{v \in \vertex(\PC)} \bigl[\tcone(\PC,v)\bigr] \lmod.
    $$
\end{theorem}

Gribanov \& Malyshev \cite{Counting_FPT_Delta} give an FPT-algorithm to compute $\FC\bigl([\PC]\bigr)$, when $\PC$ is defined by a square system $A x \leq b$ with $\det(A) \not= 0$. We refer to a refined and more effective algorithm, due to Gribanov, Shumilov, Malyshev \& Zolotykh \cite{SparseILP_Gribanov}, presented by the following Lemma. 
For completeness, we give its full proof in Appendix \ref{quad_system_th_proof}.
% At the current moment of time, the work \cite{SparseILP_Gribanov} is under review. By this reason, we give a full and self-contained proof of the lemma in Section \ref{quad_system_th_proof} of Appendix.
\begin{lemma}[Gribanov, Shumilov, Malyshev \& Zolotykh \cite{SparseILP_Gribanov}]\label{quad_system_lm}
Let $A \in \ZZ^{n \times n}$, $b \in \ZZ^n$, $\Delta = |\det(A)| > 0$. Let us consider the polyhedron $\PC = \{ x \in \RR^n \colon A x \leq b\}$. Assume that $c \in \ZZ^n$ is given, such that $\langle c, h_i \rangle > 0$, where $h_i$ are the columns of $\Delta \cdot A^{-1}$, for $i \in \intint n$. Denote $\chi = \max\limits_{i \in \intint n}\bigl\{|\langle c, h_i \rangle|\bigr\}$. Let, additionally, $S = P A Q$ be the SNF of $A$, where $P,Q \in \ZZ^{n \times n}$ are unimodular, and put $\sigma = S_{n n}$.

Then, for any $\tau >0$, the series $\fG(\PC,c;\,\tau)$ converges absolutely to a function of the type
$$
\frac{\sum\limits_{i = -n \cdot \sigma\cdot \chi}^{n \cdot \sigma\cdot \chi} \epsilon_i \cdot e^{\alpha_i \cdot \tau}}{\bigl(1 - e^{\beta_1 \cdot \tau}\bigr)\bigl(1 - e^{\beta_2 \cdot \tau}\bigr) \dots \bigl(1 - e^{\beta_n \cdot \tau}\bigr)}, 
$$ where $\epsilon_i \in \ZZ_{\geq 0}$, $\beta_i \in \ZZ_{<0}$, and $\alpha_i \in \ZZ$. This representation can be found with an algorithm, having the arithmetic complexity bound
$$
O\bigl(T_{\SNF}(n) + \Delta \cdot n^2 \cdot \sigma \cdot \chi\bigr),
$$ where $T_{SNF}(n)$ is the arithmetic complexity of computing the SNF for $n \times n$ integer matrices.

Moreover, each of the coefficients $\bigl\{\epsilon_i := \epsilon_i(g)\bigr\}$ depends only on $g = P b \bmod S \cdot \ZZ^n$, so they take at most $\Delta$ possible values, when the vector $b$ varies. All these values can be computed during the algorithm.
\end{lemma} 

In the following lemma, we summarize known facts that help to make estimates for the values of $\nu(A)$ and $\mu(A)$, which help to apply Theorem \ref{main_param_th} in different scenarios.
\begin{lemma}\label{nu_mu_lm}
Let $A \in \RR^{m \times n}$, $\rank(A) = n$, and $k = m-n$. The following relations hold for $\nu := \nu(A)$ and $\mu := \mu(A)$:
\begin{enumerate}
    \item $\nu,\mu = O\bigl(\frac{n}{k} + 1\bigr)^k$;
    \item $\nu = O\bigl(\frac{m}{n}\bigr)^{\frac{n}{2}}$, $\mu = O\bigl(\frac{m}{n}\bigr)^{\frac{n+1}{2}}$;
    \item Let $r$ and $l$ be the numbers of non-zeroes in rows and columns of $A$, respectively. Then, $\nu,\mu \leq \norm{A}_{\max}^{n} \cdot \min\{r,l\}^n$.
    %\item $\nu = O(n)^{\frac{n}{2}} \cdot \Delta^n$, $\mu = O(n)^{\frac{n+1}{2}}\cdot \Delta^{n+1}$.
\end{enumerate}
\end{lemma}
\begin{proof}
The first bounds for $\nu$ and $\mu$ follow from the trivial identities $\nu,\mu \leq \binom{m}{n} = \binom{m}{k} = O\bigl(\frac{n}{k} + 1\bigr)^{k}$. Denote by $\zeta(d,j)$ the maximum number of vertices in a polytope that is dual to the $d$-dimensional cyclic polytope with $j$ vertices. Due to the seminal work of P.~McMullen~\cite{MaxFacesTh}, we have $\nu \leq \zeta(m,n)$. Similarly, due to the seminal work of R.~Stanley \cite{UpperBoundSpheres_Stanley}, see also \cite[Corollary~2.6.5]{TriangBook_DeLoera} and \cite{CombinatoricsAndCommutative_Stanley}, $\mu \leq \zeta(m+1,n+1)-(n+1)$. Due to \cite[Section~4.7]{Grunbaum},
$$
\xi(d,j) = \begin{cases}
    \frac{j}{j-s} \binom{j-s}{s},\text{ for }d = 2s\\
    2\binom{j-s-1}{s},\text{ for }d = 2s+1\\
    \end{cases} = O\left(\frac{j}{d}\right)^{d/2}.
$$
Consequently, $\nu = O\bigl(\frac{m}{n}\bigr)^{\frac{n}{2}}$ and $\mu = O\bigl(\frac{m+1}{n+1}\bigr)^{\frac{n+1}{2}}$. Finally, the last bound for $\nu$ was proven in \cite[Lemma~4]{SparseILP_Gribanov}. The corresponding proof can be straightforwardly applied to $\mu$ without any changes.
\end{proof}

\subsection{The unbounded case, dimension of $\PC$, and rank of matrices}\label{unbouded_subs}
In the current Subsection, we are going to show that the assumptions $\rank(A) = n_x$, $\rank(B) = n_y$, and $\dim(\PC) = n_x + n_y$ (for \ref{canonical_form}) can be satisfied without any loss of generality. Additionally, we are going to show how to handle the case $\QEnum(y) = +\infty$ in our computations.
\begin{remark}[Rank of the matrix $B$]\label{B_rank_rm}
    Let us first justify the assumption $\rank(B) = n_y$, for $\PC$, defined by \ref{canonical_form}. Analysis for \ref{standard_form} is then straightforward. If $\rank(B) < n_y$, there exists an index $j$ and a nonzero vector $t \in \QQ^{n_y-1}$, such that $B_j = B_{\bar j} t$, where $\bar j = \intint{n_y} \setminus \{j\}$. Consequently, assuming that $j=1$, for $y \in \RR^{n_y}$, we have $c_{\PC}(y) = c_{\PC'}(y_1 \cdot t  + y_{\intint[2]{n_y}})$, where the polyhedron $\PC'$ is defined by a system $A x + B_{\bar 1} y \leq b$, for $x \in \RR^{n_x}$ and $y \in \RR^{n_y - 1}$. Therefore, we can work with the function $c_{\PC'} \colon \RR^{n_y-1} \to \ZZ_{\geq 0}$ instead of $c_{\PC}$. Eliminating all the linear dependencies in $B$, we can assume that $\rank(B) = n_y$, which, together with $\rank(A) = n_x$, is equivalent to the fact that $\PC$ contains no lines. Clearly, such an elimination can be performed by a polynomial-time algorithm.
\end{remark}

\begin{remark}[The Dimension of $\PC$ in \ref{canonical_form}]\label{dim_PC_rm}
Let $\PC$ be a polyhedron, defined by \ref{canonical_form}. Assume that $\dim(\PC) < n_x + n_y$. Then, there exists an index $j \in \intint m$, such that $A_j x + B_j y = b$, for any $\binom{x}{y} \in \PC$. We replace the $j$-th inequality by the inequality $A_j x + B_j y \leq b + \varepsilon$, for a sufficiently small $\varepsilon > 0$. Clearly, if the value of $\varepsilon$ is small enough, this transformation does not change the set of integer solutions and the new inequality can not hold as equality, for all $\binom{x}{y} \in \PC$, because the new polyhedron contains the old one. Note that the index $j$ with the appropriate value of $\varepsilon$ can be found by a polynomial-time algorithm. After eliminating all such $j$, we finally achieve a  polyhedron of dimension $n_x + n_y$.
\end{remark}
The next lemma shows how to satisfy the condition $\rank(A) = n_x$.
\begin{lemma}\label{lines_eliminating_lm}
Let $\PC$ be a polyhedron, defined by \ref{canonical_form}. Assume that $\rank(A) < n_x$, or equivalently, $\PC_{y}$ contain a line, for any $y \in \Pi(\PC)$. Then, there exist a set of $n_x-\rank(A)$ indices $\IC \subseteq \intint{n_x}$ and a set of values $\alpha_i \in \{-1,1\}$, for $i \in \IC$, such that the polyhedron $\PC'$, defined by
$$
    \PC' = \{x \in \PC \colon \alpha_i x_i \geq 0,\,\text{ for $i \in \IC$}\},
$$
has the following properties:
\begin{enumerate}
    \item The set $\IC$ and values $\{\alpha_i\}_{i \in \IC}$ can be found by a polynomial-time algorithm;
    \item The equality $\Pi(\PC) = \Pi(\PC')$ is true;
    \item For any $y \in \Pi(\PC)$, the polyhedron $\PC'_{y}$ contains no lines;
    \item For any $y \in \Pi(\PC)$, $\PC_{y} \cap \ZZ^{n_x} \not= \emptyset$ if and only if $\PC'_{y} \cap \ZZ^{n_x} \not= \emptyset$;
    \item Let $A'$ be the matrix, constructed by appending the rows $-\alpha_i e_i^\top$, for $i \in \IC$, to $A$. In other words, the system 
    $$
    A' x + \binom{B}{\BZero} y \leq \binom{b}{\BZero}
    $$ defines $\PC'$. Then $\mu(A') = \mu(A)$ and $\Delta(A') = \Delta(A)$.
\end{enumerate}
\end{lemma}
\begin{proof}
The set $\IC$ is defined in the following way: $\IC$ is a minimal set, such that the vectors $\{e_i\}_{i \in \IC}$ together with the column-vectors from $A^\top$ form a basis of $\RR^{n_x}$. Clearly, such a set can be found by a polynomial-time minimal basis extension algorithm. Since $r < n_x$, we have that $\IC \not= \emptyset$. Now, let us show how to construct the values  $\{\alpha_i\}_{i \in \IC}$.

    Denote $\LC := \{x \in \RR^{n_x} \colon A x = \BZero\}$ and $r := \rank(A) < n_x$. It follows that $\LC$ is exactly the lines-space of $\PC_y$, i.e. $\PC_y + \LC = \PC_y$, for any $y \in \Pi(\PC)$. Let $V \in \ZZ^{n_x \times (n_x - r)}$ be a matrix, composed of basis vectors of $\LC$. Clearly, the matrix $V$ can be constructed by a polynomial-time algorithm. We claim that, for any $i \in \IC$, there exists $j \in \intint{n_x - r}$, such that $V_{i j} \not= 0$. Definitely, if $V_{i j} = 0$, for all $j \in \intint{n_x - r}$, and some $i \in \IC$, it follows that $e_i \in \linh(A^\top)$ that contradicts to construction of $\IC$. Now, taking a sum, with sufficiently big coefficients, of columns of $V$ that correspond to such $j$, we can construct a vector $v \in \LC$, such that $v_i \not= 0$, for any $i \in \IC$. W.l.o.g, we can assume that $v$ is integer.
    For $i \in \IC$, we set 
    $$
    \alpha_i := \begin{cases}
        1, \text{ if $v_i > 0$;}\\
        -1, \text{ if $v_i < 0$}.
    \end{cases}
    $$
    Denote $\SC := \{x \in \RR^{n_x} \colon \alpha_i x_i \geq 0,\,\text{for $i \in \IC$}\}$. Clearly, the set $\LC \cap \SC$ contains the nonzero integer vector $v$ and $\PC' = \PC \cap \SC$.

    Now, the properties 1 and 3 are straightforward to see. Let us prove the $2$-nd property. Since $\PC' \subseteq \PC$, clearly $\Pi(\PC') \subseteq \Pi(\PC)$. Let us prove the opposite inclusion $\Pi(\PC) \subseteq \Pi(\PC')$. The following sequence of implications holds:
    \begin{multline*}
        y \in \Pi(\PC) \quad\Longrightarrow\quad \exists x_0 \in \RR^{n_x}\colon \binom{x_0}{y} \in \PC \quad\Longrightarrow\quad \forall \tau \in \RR \colon \binom{x_0 + \tau \cdot v}{y} \in \PC\\
        \Longrightarrow\quad \exists \tau \in \RR_{>0} : \binom{x_0 + \tau \cdot v}{ y } \in \PC \cap \SC \quad\Longrightarrow\quad y \in \Pi(\PC'),
    \end{multline*} which proves the inclusion.

    Let us prove the property 4, and fix $y \in \Pi(\PC)$. Trivially, if $x \in \PC'_{y} \cap \ZZ^{n_x}$, then $x \in \PC_{y} \cap \ZZ^{n_x}$. Let us prove the opposite implication. The following sequence of implications holds:
    \begin{multline*}
        \exists x_0 \in \PC_{y} \cap \ZZ^{n_x} \quad\Longrightarrow\quad \forall \tau \in \ZZ\colon x_0 + \tau v \in \PC_{y} \cap \ZZ^{n_x} \quad\Longrightarrow\\
        \Longrightarrow\quad \exists \tau \in \ZZ_{>0} \colon x_0 + \tau v \in \PC \cap \ZZ^{n_x} \cap \SC \quad\Longrightarrow\quad \PC' \cap \ZZ^{n_x} \not= \emptyset,
    \end{multline*} which proves the implication.

    Finally, let us prove the property 5. Let $\CCal = \cone(A^\top)$ and $\CCal' = \cone(A'^\top)$. Note that $\dim(\CCal) = r < n = \dim(\CCal')$. Following to Definition \ref{triang_def}, let $\TS$ be any triangulation of $\CCal$. The corresponding triangulation $\TS'$ of $A'$ can be constructed in the following way: for any simple cone $\TC \in \TS$, we construct $\TC' \in \TS'$ by adding the rays $\{\alpha_i e_i\}_{i \in \IC}$ to the set of generating rays of $\TC$. The dimension of the new simple conses $\TC' \in \TS'$ will be equal to $n_x$, and they will form a triangulation of $\CCal'$. Since each $n_x$-dimensional simple cone, whose rows are composed of the columns of $A'^\top$, must contain $\{\alpha_i e_i\}_{i \in \IC}$ as the generating rays, it follows that each triangulation $\TS'$ of $\CCal'$ can be constructed by the presented algorithm. So, there is a bijection between the triangulations of $\CCal$ and $\CCal'$, which have the same sizes. Consequently, $\mu(A) = \mu(A')$. The equality $\Delta(A') = \Delta(A)$ is trivial.
\end{proof}

In the assumption that $\rank(A) = n$, the next lemma gives a way to handle the counting problem in an unbounded polyhedron.
\begin{lemma}\label{param_bounding_lm}
    Let $\PC$ be a polyhedron, defined by \ref{canonical_form}, and $\rank(A) = n_x$. Assume that $\PC_{y}$ is unbounded, for some $y \in \Pi(\PC)$. Then, there exists a function $g \colon \QQ^{n_y} \to \QQ$, computing in polynomial time, and a polyhedron $\PC'$, defined by a system
    $$
    \begin{cases}
        A' x + B' y \leq b'\\
        x \in \RR^{n_x},\, y \in \RR^{n_y+1}
    \end{cases}
    $$ with $A' \in \ZZ^{(m+1) \times n_x}$, $B' \in \ZZ^{(m+1)\times (n_y+1)}$, and $b' \in \QQ^{m+1}$, such that the following properties hold:
    \begin{enumerate}
    \item The polyhedron $\PC'$ and the function $g$ can be constructed by a polynomial-time algorithm;
    
    \item The inequality $\Delta(A') \leq n_x \cdot \Delta(A)$ is true;

    \item For any $y \in \Pi(\PC')$, $\PC'_{y}$ is bounded;

    \item For any $y \in \Pi(\PC)$, we have
    $$
    \QEnum(y) \not= 0 \quad \Longleftrightarrow \quad \EC_{\PC'}(y') \not= 0,
    $$
    where $y' = \dbinom{y}{g(y)}$.
    \end{enumerate}
\end{lemma}
\begin{proof}
Due to \cite{BoroshTreybigProof}, if an arbitrary system $M x \leq h$, where $M \in \ZZ^{m \times n_x}$ and $h \in \ZZ^m$, has an integer solution, then there exists an integer solution $z$, such that $\|z\|_{\infty} \leq \Delta\bigl((M\,h)\bigr)$, where $(M\,h)$ is the system extended matrix. Denote $b_{y} = \lfloor b - B y \rfloor$. Therefore, if $\PC_{y} \cap \ZZ^{n_x} \not= \emptyset$, then there exists $z \in \PC_{y} \cap \ZZ^n$, such that $\|z\|_{\infty} \leq \Delta\bigl((A\,b_{y})\bigr) \leq (n_x)^{n_x/2} \cdot \max\{\|A\|_{\max},\|b_{y}\|_{\infty}\}^{n_x}$. 
% Since $\bigl|\lfloor x \rfloor\bigr| \leq |x| + 1$, we have 
% $$
%  \|b_{y}\|_{\infty} \leq \|b\|_{\infty} + \|B y\|_{\infty} + 1 
%  \leq \|b\|_{\infty} + \|B\|_{\max} \|y\|_{1} + 1.
% $$
Since, for some $y \in \Pi(\PC)$, $\PC_{y}$ is unbounded, the cone $\CCal = \{x \in \RR^n \colon A x \leq \BZero\}$ is non-zero and, consequently, for any $y \in \Pi(\PC)$, $\PC_{y}$ is unbounded. Since $\rank(A) = n$, $\PC_{y}$ contains no lines, so $\CCal$ is pointed. Let $A_{\BC}$ be some basis sub-matrix of $A$ and $c$ be the sum of columns of $-A_{\BC}^\top$. Clearly, $\CCal \cap \{x \in \RR^{n_x} \colon c^\top x \leq c_0\} = \{0\}$, for any $c_0 \in \RR_{\geq 0}$.

Now, we define the polyhedron $\PC'$, appending the inequality $c^\top x \leq y_{n_x+1}$ to the system $A x + B y \leq b$, where $y_{n_x+1} \in \RR$ is a new parametric variable. Clearly, the properties 2 and 3 are satisfied for $\PC'$ and $A'$. Finally, define $g(y) := (n_x)^{n_x/2} \cdot \|c\|_{1} \cdot \max\{\|A\|_{\max},\|b_{y}\|_{\infty}\}^{n_x}$. Due to the proposed reasoning, the properties 1 and 4 are also satisfied.
\end{proof}

\section{Construction of the chamber decomposition and a proof of Theorem \ref{chamber_decomp_th}}\label{chamber_decomp_proof}

In this Section, we give the proof of Theorem \ref{chamber_decomp_th} and develop an algorithm to construct a chamber decomposition of $\PC$. The paper \cite[Section~3]{Parametric_Clauss} of Clauss \& Loechner gives an algorithm to construct a collection $\DS$ of \emph{full-dimensional chambers}, such that
\begin{enumerate}
    \item The equality $\Pi(\PC) = \bigcup\limits_{\DC \in \DS} \DC$ is true;
    \item The inequality $\dim(\DC_1 \cap \DC_1) < n_y$ is true, for any $\DC_1,\DC_2 \in \DS$, with $\DC_1 \not= \DC_2$;
    \item For any $\DC \in \DS$ and $y \in \DC$, the polytopes $\{\PC_{y}\}$ have a fixed collection of parametric vertices, given by affine transformations of $y$.
\end{enumerate}
Note that, for $y$ in a boundary of some full-dimensional chamber $\DC \in \DS$, the corresponding parametric vertices may stick together to a single point in $\RR^{n_x}$. So, the polyhedron $\PC_y$ may not have the same combinatorial type for different points $y \in \DC$. 
%All the polytopes from $\QS$ are also called \emph{chambers}. Here, we will call them as \emph{full-dimensional chambers}. 
Following the proof of \cite[Lemma~3]{CountingInParametricPolyhedra}, let us consider hyperplanes in the parametric space $\RR^{n_y}$, formed by the affine hulls of $(n_y-1)$-dimensional faces (facets) of full-dimensional chambers $\DC \in \DS$. Let $\HS$ be the set of all such hyperplanes. 

Due to \cite{HyperplaneSpacePartition} (see also \cite[Lemma~3.5]{CountingFunctionEncoding}), the hyperplanes from $\HS$ divide the parameter space into at most $O\bigl(|\HS|^{n_y}\bigr)$ cells. Due to construction of the full-dimensional chambers (see \cite[Section~3]{Parametric_Clauss} or the proof of \cite[Lemma~3]{CountingInParametricPolyhedra}), these hyperplanes correspond to the projections of the generic $(n_y-1)$-dimensional faces of $\PC$ into the parametric space $\RR^{n_y}$. Since a part of the cells forms a subdivision of the chambers and since $|\HS| \leq f_{n_y-1}$, the total number of full-dimensional chambers can be bounded by $O\bigl((f_{n_y-1})^{n_y}\bigr)$.

Clauss \& Loechner \cite{Parametric_Clauss} give an algorithm that computes the collection $\DS$, which, for each $\DC \in \DS$, also computes a finite set $\TS_{\DC}$ of affine functions $\TC \in \TS_{\DC}$, $\TC \colon \RR^{n_y} \to \RR^{n_x}$, such that all the parametric vertices of $\PC_{y}$ are given by $\bigl\{\TC(y)\colon \TC \in \TS_{\DC}\bigr\}$. The complexity of Clauss \& Loechner's algorithm is bounded by the number of chambers times the total number of parametric vertices. Due to \cite{VerticesOfParametricPolyhedra}, the parametric vertices correspond to the $n_y$-dimensional faces of $\PC$, so their number is bounded by $O(f_{n_y})$. Therefore,  Clauss and Loechner's algorithm needs 
$$
(f_{n_y-1})^{n_y} \cdot f_{n_y} \cdot \poly(n_x, n_y, m) \qquad \text{operations.}
$$

For any fixed $\DC \in \DS$ and $y \in \inter(\DC)$, the vertices $
\TC(y)\colon \TC \in \TS_{\DC}$ are unique. But, for $y \in \DC_1 \cap \DC_2$, where $\DC_1 \not= \DC_2$, some vertices can coincide. Due to \cite{CountingInParametricPolyhedra,CountingFunctionEncoding}, this problem can be resolved by working with a wider class of chambers that consists of $\DS$ and all the faces of chambers from $\DS$. This new collection of chambers is exactly given by Definition \ref{chamber_decomp_def}, which is denoted by $\QS$.

% Due to \cite{CountingInParametricPolyhedra,CountingFunctionEncoding}, we have
% \begin{enumerate}
%     \item $\Pi(\PC) = \bigcup\limits_{\QC \in \QS} \relint(\QC)$;
%     \item $\relint(\QC_1) \cap \relint(\QC_2) = \emptyset$, for any $\QC_1,\QC_2 \in \QS$, with $\QC_1 \not= \QC_2$;
%     \item for any $\QC \in \QS$ and any $y \in \relint(\QC)$, the polytope $\PC_{y}$ has a fixed combinatorial type. Denote the set of unique vertices of $\PC_{y}$ by $\{\TC(y)\colon \TC \in \TS_{\QC}\}$, where $\TS_{\QC}$ is a finite set of affine functions $\TC \colon \RR^{n_y} \to \RR^{n_x}$.
% \end{enumerate}
% We call the polyhedra from the new collection $\QS$ as the \emph{low-dimensional chambers} or just \emph{chambers}. Note that $\QS \subseteq \QS$. 
Our goal is to find the collection $\QS$ with the corresponding lists of the unique parametric vertices. Very briefly, we enumerate all the possible affine sub-spaces, induced by all the possible intersections of the hyperplanes from $\HS$. For any such an affine subspace $\LC \subseteq \RR^{n_y}$ and a full-dimensional chamber $\DC \in \DS$, we consider the intersection $\LC \cap \DC$. If this intersection forms a non-empty face of $\DC$, then we can declare that we have found some low-dimensional chamber from $\QS$. Note that this check can be performed by a polynomial-time algorithm. There are three main difficulties:
\begin{enumerate}
    \item[1)] Two different full-dimensional chambers $\DC_1$ and $\DC_2$ can have the same intersection with $\LC$: $\DC_1 \cap \LC = \DC_2 \cap \LC$; 
    \item[2)] If $\LC' \subset \LC$ is an affine subspace of $\LC$, then it is possible that $\DC \cap \LC = \DC \cap \LC'$, for some full-dimensional chamber $\DC \in \DS$. In other words, the dimension of $\LC \cap \DC$ is strictly less than $\dim(\LC)$;
    \item[3)] As it was already noted, if $y \in \DC_1 \cap \DC_2$, then some vertices, given by different affine functions from $\TS_{\DC_1}$, can coincide. In other words, $\TC_1(y) = \TC_2(y)$, for different $\TC_1, \TC_2 \in \TS_{\DC_1}$ and $y \in \DC_1 \cap \DC_2$. How we are able to find such duplicates?
\end{enumerate}

Let us first deal with the difficulties 1) and 3). Let us fix $\LC$, a full-dimensional chamber $\DC \in \DS$, and the collection of parametric vertices $\TS_{\DC}$ and consider the lower-dimensional chamber $\QC = \LC \cap \DC$. Assume that $\dim(\LC) = \dim(\QC)$, we will break this assumption later. There exists a matrix $B \in \ZZ^{n_y \times d_{\LC}}$, where $d_{\LC} = \dim(\LC)$, such that $\LC = \linh(B)$. Suppose that $\TC(y) \in \TS_{\DC}$ is given by $\TC(y) = T y + t$, where $T \in \QQ^{n_x \times n_y}$ is a rational transformation matrix and $t \in \QQ^{n_x}$ is a rational translation vector. Let us show, how to find duplicates in the list $\TS_{\DC}$ of parametric vertices, for $y \in \relint(\QC)$. If $\TC_1$ and $\TC_2$ form a duplicate, then
$$
T_1 y + t_1 = T_2 y + t_2,\quad \text{for any } y \in \relint(\QC),
$$ that is equivalent to
$$
(T_1-T_2) y = t_2 - t_1,\quad \text{for any } y \in \affh(\QC).
$$
Due to the assumption $\dim(\LC) = \dim(\LC \cap \DC) = \dim(\QC)$, we have
$$
(T_1 - T_2) B x = t_2 - t_1,\quad \text{for any } x \in \RR^{d_{\LC}}.
$$ Since the solutions set of the last system is $d_{\LC}$-dimensional, we have $\rank\bigl((T_2 - T_1) B\bigr) = 0$. Consequently, $T_1 B = T_2 B$ and $t_1 = t_2$, so the matrices $\{T B\}$ and vectors $\{t\}$ must serve as a unique representation of affine functions $\TC \in \TS_{\DC}$, for a fixed subspace $\LC$. Consequently, we need to compute the basis $B$ of $\LC$, compute the pairs $\{(T B, t)\}$, sort the resulted list, and delete all the duplicates. Since $|\TS_{\DC}| \leq \nu$, this work can be done with $$
\nu \cdot \log (\nu) \cdot \poly(n_x, n_y, m)
$$ operations. Since $|\DS| = O\bigl((f_{n_y-1})^{n_y}\bigr)$ and, due to Lemma \ref{nu_mu_lm}, $\log(\nu) = O\bigl(\poly(n_x, m)\bigr)$, together with enumerating of all the full-dimensional chambers $\DC \in \DS$ ($\LC$ is fixed), it gives the arithmetic complexity bound:
$$
(f_{n_y-1})^{n_y} \cdot \nu \cdot \poly(n_x, n_y, m).
$$

Now, let us show how to resolve the difficulty number 1) for the low-dimensional chambers $\QC = \DC \cap \LC$ with $\dim(\QC) = \dim(\LC)$. Any such a chamber is uniquely represented by the list of unique pairs
$$
\bigl\{(TB, t)\colon \TC(y) = T y + t,\, \TC \in \TS_{\DC}\bigr\}.
$$ So, we can consider this list as a unique identifier of $\QC \in \QS$ and all the duplicated chambers can be eliminated just by sorting. Since the length of each list is bounded by $\nu$ and $|\DS| = O\bigl((f_{n_y-1})^{n_y}\bigr)$, the total arithmetic complexity of this procedure is again
$$
(f_{n_y-1})^{n_y} \cdot \nu \cdot \poly(n_x, n_y, m).
$$

Let us simultaneously discuss, how to resolve the difficulty number 2) and how to break the assumption $\dim(\LC) = \dim(\LC \cap \DC)$, for $\DC \in \DS$. To do that, we need first to sketch the full algorithm that constructs the chamber decomposition $\QS$. First, we use P.~Clauss \& V.~Loechner's algorithm to construct the full-dimensional chambers $\DS$ together with the corresponding parametric vertices. Next, we enumerate all the affine sub-spaces, induced by intersections of hyperplanes from $\HS$. The enumeration follows to the following partial order: if $\LC'$ and $\LC$ are affine sub-spaces, such that $\LC' \subset \LC$, then $\LC'$ will be processed first. So, we start from $0$-dimensional sub-spaces, corresponding to intersections of $n_y$ linearly independent hyperplanes from $\HS$, and, using the induction principle, move forward from $(k-1)$-dimensional chambers to $k$-dimensional chambers. For a $0$-dimensional affine subspace $\LC$ (which is a point) and a full-dimensional chamber $\DC \in \DS$, the assumption $\dim(\LC) = \dim(\LC \cap \DC)$ naturally holds. Since the considered assumption holds for $0$-dimensional $\LC$, we can find all the $0$-dimensional chambers from the collection $\QS$ with the sets of their unique parametric vertices, using the method, discussed earlier.

Let $k \geq 1$ and suppose inductively that we want process $k$-dimensional subspace $\LC$, when the sub-spaces $\LC' \subset \LC$, with $\dim(\LC') = k-1$, are already processed. When we say “processed”, we mean that all the unique low-dimensional chambers $\LC' \cap \DC$, for $\DC \in \DS$, with $\dim(\LC' \cap \DC) = k-1$ are already constructed. We recall all the full-dimensional chambers $\DC \in \DS$, such that $\DC \cap \LC'$ forms a $(k-1)$-dimensional face of $\DC$ that were computed previously. If $\dim(\DC \cap \LC) = k-1$, then we dismiss the chamber $\DC \cap \LC$, because it just coincides with some $\DC \cap \LC'$. In the opposite case, when $\dim(\DC \cap \LC) = k = \dim(\LC)$ and $\DC \cap \LC$ forms a face of $\DC$ that was not considered before, we can use the approach presented earlier to construct the unique set of the parametric vertices for $\DC \cap \LC$ and put the resulting chamber into the collection $\QS$. So, we need to learn how to distinguish between the cases $\dim(\DC \cap \LC) = k$ and $\dim(\DC \cap \LC) = k-1$. 

Let $B \in \ZZ^{n_y \times k}$ and $B' \in \ZZ^{n_y \times (k-1)}$ be bases of $\LC$ and $\LC'$, chosen, such that $B = \bigl(B'\; h\bigr)$, for $h \in \ZZ^{n_y}$. The equality $\dim(\DC \cap \LC) = k$ holds if and only if there exists affine functions $\TC_1(y) = T_1 y + t_1$ and $\TC_2(y) = T_2 y + t_2$ from the set 
 $\TS_{\DC}$, such that
\begin{gather*}
    \forall y \in \DC \cap \LC' \colon \quad \TC_1(y) = \TC_2(y) \text{ and }\\
    \exists y \in \DC \cap \LC \colon \quad \TC_1(y) \not= \TC_2(y),
\end{gather*} which is equivalent to
\begin{gather*}
    \forall y \in \affh(\DC \cap \LC') \colon \quad \TC_1(y) = \TC_2(y) \text{ and }\\
    \exists y \in \affh(\DC \cap \LC) \colon \quad \TC_1(y) \not= \TC_2(y).
\end{gather*} 
The last is equivalent to 
\begin{gather*}
    \forall x \in \RR^{k-1} \colon \quad \bigl(T_1 - T_2\bigr) B' x = t_2 - t_1 \text{ and }\\
    \exists x \in \RR^{k} \colon \quad \bigl(T_1 - T_2\bigr)\bigl(B'\; h\bigr) x \not= t_2 - t_1.
\end{gather*}
By the same reasoning, $\rank\bigl((T_1 - T_2) B'\bigr) = 0$, $T_1 B' =  T_2 B'$ and $t_2 = t_1$. Hence, the second property can be satisfied only if $T_1 h \not= T_2 h$. So, we have the following criterion:
\begin{equation}\label{k_dim_criteria}
    \dim(\DC \cap \LC) = k \quad \Longleftrightarrow \quad \exists \TC_1,\TC_2 \in \TS_{\DC} \colon\; \TC_1(h) \not= \TC_2(h).
\end{equation}
Let us summarize the whole algorithm:
\begin{enumerate}
    \item {\bf Construction of full-dimensional chambers.} Using P.~Clauss \& V.~Loechner's algorithm, we construct the collection $\DS$ of  full-dimensional chambers together with the sets of affine functions $\TS_{\DC}$, for each $\DC \in \DS$. It takes $(f_{n_y - 1})^{n_y} \cdot f_{n_y} \cdot \poly(n_x, n_y, m)$ arithmetic operations.
    \item {\bf Construction of lower-dimensional chambers. } To construct the collection $\QS$, we consider affine sub-spaces $\LC$, induced by all the possible intersections of hyperplanes $\HS$ with full-dimensional chambers from $\DS$.

    \medskip
    
    Assume that all the unique $(k-1)$-dimensional chambers of the type $\DC \cap \LC'$, where $\LC'$ is the $(k-1)$-dimensional intersection of hyperplanes from $\HS$ and $\DC \in \DS$, have already been constructed with their unique sets of parametric vertices. 
    For all the $k$-dimensional intersections $\LC$ and for all $\DC \in \DS$, we perform the following operations:
    % For all the $k$-dimensional intersections $\LC$, with $\LC' \subset \LC$, and for all $\DC \in \DS$, with $\LC' \cap \DC$ forming a $(k-1)$-dimensional face of $\DC$, we perform the following operations:
    \item {\bf Dimension check. } Let $\QC = \DC \cap \LC$. Using the criteria \eqref{k_dim_criteria}, we check that $\dim(\QC) = k$. If it does not hold, then we skip $\QC$;
    \item {\bf $\QC$-face check. } We check that $\QC$ is a face of $\DC$, which can be done by a polynomial-time algorithm.  If it does not hold, then we skip $\QC$;
    \item {\bf Erase duplicated parametric vertices.} Using the algorithm, presented earlier, we erase all the duplicated parametric vertices of $\QC$ and append $\QC$ to the collection $\QS$.
    \item {\bf Erase duplicated chambers.} After all $\DC \in \DS$ (for a fixed $\LC$) are processed, we remove the duplicated chambers $\QC = \DC \cap \LC$ from $\QS$ by the algorithm, mentioned earlier.
\end{enumerate}
For any fixed $\LC$, the complexity of the steps 3-6 is bounded by 
$$
(f_{n_y-1})^{n_y} \cdot \nu \cdot \poly(n_x, n_y, m).
$$
Since $\abs{\HS} \leq f_{n_y - 1}$, the number of different affine sub-spaces $\LC$, induced by intersections of hyperplanes from $\HS$, is bounded by
$$
\sum\limits_{i = 0}^{n_y} \binom{f_{n_y - 1}}{i} = O\bigl((f_{n_y-1})^{n_y}\bigr).
$$
Clearly, they can be enumerated with  $(f_{n_y-1})^{n_y}\cdot n_y^{O(1)}$ operations. Consequently, the chamber decomposition can be constructed with 
$$
\Bigl((f_{n_y-1})^{n_y} \cdot f_{n_y} + (f_{n_y-1})^{2n_y} \cdot \nu\Bigr)\cdot \poly(n_x, n_y, m)
$$ operations. The total number of chambers in the decomposition $\QS$ can be bounded by the number of ways to intersect any of the full-dimensional chambers with different affine sub-spaces $\LC$. Therefore, $\abs{\QS} = O\bigl((f_{n_y - 1})^{2 n_y}\bigr)$.

Let us estimate the complexity of a query to find $\QC \in \QS$ with $y \in \relint(\QC)$, for a given $y \in \QQ^{n_y}$. Note that each chamber $\QC$ naturally corresponds to a vector from $\{-1,0,1\}^{\abs{\HS}}$. To support our queries, we construct a hash-table that maps vectors from $\{-1,0,1\}^{\abs{\HS}}$ to $\QS$. The expected complexity to construct such a hash-table is bounded by the number of chambers timed the size of the vector that represents a chamber, which is $O\bigl(\abs{\QS} \cdot \abs{\HC} \bigr) = O\bigl((f_{n_y - 1})^{2 n_y+1}\bigr)$. Consequently, to perform a query we substitute an input vector $y$ to the inequalities corresponding to the hyperplanes $\HS$, and map $y$ to $\{-1,0,1\}^{\abs{\HS}}$. Using the hash-table, we find the corresponding chamber $\QC$ with $y \in \relint(\QC)$. Therefore, each query costs of $O\bigl(n_y \cdot \abs{\HS}\bigr) = O(n_y \cdot f_{n_y-1})$ operations. The total arithmetic complexity together with a construction of the hash-table is $O^*\bigl((f_{n_y-1})^{n_y} \cdot f_{n_y} + (f_{n_y-1})^{2n_y} \cdot (\nu + f_{n_y-1} )\bigr)$, which finishes the proof.

\section{Proof of the main theorem}
\label{main_param_th_proof}

First, let us make some basic preliminary analysis of degenerate situations. Since, for some $y \in \RR^{n_y}$, $\PC_{y}$ is bounded, it contains no lines. Consequently, $\rank(A) = n_x$. Due to Remark \ref{B_rank_rm}, w.l.o.g. we can assume that $\rank(B) = n_y$. 
% Together with $\rank(A) = n_x$, it grants that $\PC$ contains no lines. 
Due to Remark \ref{dim_PC_rm}, we can assume that $\dim(\PC) = n_x + n_y$. Summarizing the preliminary analysis, we have the following:
\begin{enumerate}
    \item $\dim(\PC) = n_x + n_y$ and $\dim\bigl(\Pi(\PC)\bigr) = n_y$;
    \item $\rank(A) = n_x$, $\rank(B) = n_y$, and $\PC$ contains no lines.
\end{enumerate}

As the first step of the preprocessing algorithm, we construct a chamber decomposition $\QS$ of $\PC$, using Theorem \ref{chamber_decomp_th}. In the next stage of the preprocessing algorithm, we deal independently with each chamber $\QC \in \QS$.

\subsection{Dealing with a fixed chamber $\QC \in \QS$}\label{fixed_cham_subs}

Let us fix a chamber $\QC$, chosen from the collection $\QS$, which was constructed in the previous stage. As it was noted before, for any $y \in \relint(\QC)$, polytopes $\{\PC_{y}\}$ have a fixed combinatorial type and a fixed set of unique parametric vertices $\pvertex(\QC) = \{\VC_{\BC} \colon \BC \in \base(\QC)\}$. For the sake of simplicity, denote $n := n_x$.

Consider first the case, when $\dim(\PC_{y}) < n$, for all $y \in \relint(\QC)$. Choose a point $y' \in \relint(\QC)$, it can be done by a polynomial number of operations, and consider $\PC_{y'}$. The polytope $\PC_{y'}$ is defined by a system $A x \leq b'$, where $b' = b + B y'$. Since $\dim(\PC_{y'}) < n$, there exists an index $j \in \intint m$, such that $A_j x = b'_j$, for any $x \in \PC_{y'}$. Note that such $j$ could be found by a polynomial number of operations. W.l.o.g., assume that $j = 1$ and $\gcd(A_1) = 1$. Let $A_1 = H Q$, where $H \in \ZZ^{1 \times n}$ is the \emph{Hermite Normal Form (HNF)} of $A_1$ and $Q \in \ZZ^{n \times n}$ is unimodular. Since $\gcd(A_1) =1$, we have $H = (1\,\BZero_{n-1})$. After the unimodular map $x' = Qx$, the system $A x \leq b'$ transforms to the integrally equivalent system 
$$
\begin{pmatrix}
    1 & \BZero_{n-1}\\
    h & A'\\
\end{pmatrix} x \leq b',
$$ where $h \in \ZZ^{m-1}$ and $A' \in \ZZ^{(m-1)\times(n-1)}$. Note that $\Delta(A') = \Delta(A) = \Delta$. Since the first inequality always holds as an equality on the solutions set, we can just substitute $x_1 = b'_1$. As the result, we achieve a new integrally equivalent system with $n-1$ variables: $A' x \leq b'_{\intint[2]{m}} - b'_1 \cdot h$. Since all these steps are independent on a choice of $y \in \relint(\QC)$, we can think $b'$ as a function $b'(y)$. Consequently, we can replace the polytope $\PC_{y}$ by a new polytope
\begin{multline}\label{reduced_dim_Py_eq}
A' x \leq b'_{\intint[2]{m}}(y) - b'_1(y) \cdot h = \\ 
= B_{\intint[2]{m}} y + b_{\intint[2]{m}} - h \cdot (B_1 y + b_1) = \\
= (B_{\intint[2]{m}} - h B_1) \cdot y + (b_{\intint[2]{m}} - b_1 h)
\end{multline}
with only $n-1$ variables. The set of parametric vertices of the new polytope can be constructed by the following way (noting that the all of them are unique). Let $\VC(y)$ be some parametric vertex of $\PC_{y}$, it corresponds to some base $\BC$ of $A$. In other words, $A_{\BC} \VC(y) = b'(y)$, for any $y \in \relint(\QC)$. Clearly, the related parametric vertex $\VC'(y)$ of the new polytope corresponds to the base $\BC' = \BC \setminus \{1\}$. Hence, it can be found by the formula 
$$
\VC'(y) = \bigl( A'_{\BC'} \bigr)^{-1} \cdot \bigl( b'_{\BC'}(y) - b'_1(y) \cdot h \bigr),
$$ due to \eqref{reduced_dim_Py_eq}, it is an affine map $\RR^{n_y} \to \RR^{n - 1}$.

Now, due to the proposed reasoning, we can assume that $\dim(\PC_{y}) = n$, for any $y \in \relint(\QC)$. Due to M.~Brion's Theorem \ref{Brion_th}, we have
\begin{multline*}
    [\PC_{y}] \equiv \sum\limits_{\VC \in \pvertex(\QC)} \bigl[\tcone\bigl(\PC_{y},\VC(y)\bigr)\bigr] \equiv \\
    \equiv \sum\limits_{\VC \in \pvertex(\QC)} \bigl[\VC(y) + \fcone\bigl(\PC_{y},\VC(y)\bigr)\bigr] \lmod.
\end{multline*}

Let us fix a vertex $\VC(y)$ and consider the cone $\CCal = \fcone\bigl(\PC_{y},\VC(y)\bigr)$. Denote $\JC := \JC\bigl(\VC(y)\bigr) = \{j \colon A_j \VC(y) = b_j\}$. Clearly, $\CCal = \PC\bigl(A_{\JC}, \BZero\bigr)$, and, consequently, $\CCal^\circ = \cone\bigl(A_{\JC}^\top\bigr)$. We apply the triangulation of $\CCal^\circ$ into simple cones formed by some sub-set of bases of $A^\top_{\JC}$ (see Definition \ref{triang_def}). Let $q_{\VC}$ be the total number of simple cones in the triangulation, and let $\BS_{\VC}$, with $|\BS_{\VC}| = q_{\VC}$, be the corresponding sub-set of bases. In other words, we have
$$
[\CCal^\circ] \equiv \bigcup_{\BC \in \BS_{\VC}} \bigl[\cone(A^\top_{\BC})\bigr] \ldcmod.
$$

\begin{remark}\label{triang_complexity_rm}
It is clear that constructing a triangulation for a pointed cone in $\RR^{n}$ is equivalent to constructing a triangulation for a point configuration in $\RR^{n-1}$. Therefore, due to \cite[Lemma~8.2.2]{TriangBook_DeLoera} (see also \cite{Triangulation_Gruzdev}, where another algorithm is proposed), the triangulation can be computed with $O(q_{\TC} \cdot n^3)$ operations.    
\end{remark}

Next, we use the \emph{duality trick}, see \cite[Remark~4.3]{BarvPom}. Due to \cite[Theorem~5.3]{BarvBook} (see also \cite[Theorem~2.7]{BarvBook}), there is a unique linear transformation $\mathfrak{D} \colon \PS(\VC) \to \PS(\VC)$, for which $\mathfrak{D}\bigl([\PC]\bigr) = [\PC^{\circ}]$. Consequently, due to Remark \ref{fcone_generator_rm}, we have
$$
[\CCal] \equiv \bigcup_{\BC \in \BS_{\VC}} \bigl[\cone(A^\top_{\BC})^\circ\bigr] \equiv \bigcup_{\BC \in \BS_{\VC}} \bigl[\PC(A_{\BC}, \BZero)\bigr] \lmod,
$$ and, due to Theorem \ref{linear_map_eval_th}, we have
\begin{multline*}
\bigl[\VC(y) + \CCal\bigr] \equiv \bigcup_{\BC \in \BS_{\VC}} \bigl[\VC(y) + \PC(A_{\BC}, \BZero)\bigr] \equiv \\
\equiv \bigcup_{\BC \in \BS_{\VC}} \bigl[\PC\bigl(A_{\BC}, A_{\BC} \VC(y) \bigr)\bigr] \lmod.    
\end{multline*}
The triangulation of all the cones $\cone(A^\top_{\JC})$ induces the triangulation of the whole normal fun $\cone(A^\top)$. Hence, $\sum_{\VC \in \pvertex(\QC)} q_{\VC} \leq \mu$. Consequently, combining $\BS_{\VC}$, for different $\VC \in \pvertex(\QC)$, in one big set $\BS$, we have
$$
[\PC_{y}] \equiv \sum\limits_{\BC \in \BS} \Bigl[\PC\bigl(A_{\BC}, A_{\BC} \VC_{\BC}(y)\bigr)\Bigr] \lmod,
$$ where $|\BS| \leq \mu$ and some $\VC_{\BC}$ can be equivalent, for different $\BC \in \BS$. Note that, since $\sum_{\VC \in \pvertex(\QC)} q_{\VC} \leq \mu$, and, due to Remark \ref{triang_complexity_rm}, this decomposition can be computed with $O(\mu \cdot n^3)$ operations.

Define a set $\EC$ of \emph{edge-directions} in the following way. For any element $\cone(M)$ of the resultant triangulation of $\cone(A^\top)$, we put all the columns of $\det(M) \cdot M^{-1}$ into $\EC$. If a pair of elements in $\EC$ differs only by a sign, we remove some element of the pair. Note that elements of the resulting set $\EC$ correspond exactly to the direction-vectors of the edges of $\PC_y$, assuming that the directions do not matter. 

Let us assume that a vector $c \in \ZZ^n$ is chosen, such that $c^\top h \not= 0$, for each $h \in \EC$, and denote $\chi = \max\limits_{h \in \EC} \bigl\{|c^\top h|\bigr\}$. Assume additionally that $\chi = \Omega(n)$. Denote $f(\PC;\tau) = \FC\bigl([\PC]\bigr)(\tau c)$, for any rational polyhedra $\PC$, where $\FC$ is the evaluation, considered in Theorem \ref{Pukhlikov_th}. Note that, for any $\BC \in \BS$, $$
f\Bigl(\PC\bigl(A_{\BC}, A_{\BC} \VC_{\BC}(y)\bigr);\tau\Bigr) = f\Bigl(\PC\bigl(A_{\BC}, \TC_{\BC}(y) \bigr);\tau\Bigr),
$$ where $\TC_{\BC}(y) := \lfloor A_{\BC} \VC_{\BC}(y) \rfloor$ is an affine step-function. Denote $f_{\BC}(y;\tau) := f\Bigl( \PC\bigl(A_{\BC}, \TC_{\BC}(y)\bigr); \tau \Bigr)$. Due to Theorem \ref{Pukhlikov_th}, we can write
$$
f\bigl(\PC_{y} ; \tau\bigr) = \sum\limits_{\BC \in \BS} f_{\BC}(y;\tau).
$$

Let us fix $\BC \in \BS$ and denote $f(y;\tau) := f_{\BC}(y;\tau)$, $\TC := \TC_{\BC}$, $\delta := |\det(A_{\BC})|$, and let $(h_1, \dots, h_n)$ be the columns of $\delta \cdot (A_{\BC})^{-1}$. Let, additionally, $S = P A_{\BC} Q$ be the SNF of $A_{\BC}$, for unimodular matrices $P,Q \in \ZZ^{n \times n}$, and $\sigma := S_{n n}$. Denote the orders of $(h_1, \dots, h_n)$ modulo $S \ZZ^n$ by $(r_1,\dots,r_n)$. 
Due to the assumption on the vector $c$, it satisfies the conditions of Lemma \ref{quad_system_lm}, applied to the polyhedron $\PC\bigl(A_{\BC}, \TC_{\BC}(y)\bigr)$. Due to Lemma \ref{quad_system_lm},
\begin{equation}\label{single_vertex_repr_eq}
f(y;\tau) = e^{\langle c, A_{\BC}^{-1} \TC(y) \rangle \tau} \cdot \frac{\sum\limits_{i = - n \cdot \sigma\cdot \chi}^{n \cdot \sigma\cdot \chi} \epsilon_i\bigl(g(y)\bigr) \cdot e^{-\frac{i}{\delta} \tau}}{\bigl(1 - e^{-\beta_1 \cdot \tau}\bigr) \dots \bigl(1 - e^{-\beta_n \cdot \tau}\bigr)},    
\end{equation}
where $\beta_i = \langle c, \frac{r_i}{\delta} h_i \rangle$, $g(y) = P \TC(y) \bmod S \cdot \ZZ^n$, and $\epsilon_i \colon S \ZZ^n \to \ZZ_{\geq 0}$.

Now, we are interested in the constant term in the Tailor's decomposition of $f(y;\tau)$. We have,
\begin{gather*}
    e^{- \frac{i}{\delta} \tau} = \sum\limits_{k = 0}^{+\infty} \tau^k \cdot \frac{(-i)^k}{\delta^k k!}, \\
    \sum\limits_{i = - n \cdot \sigma \cdot \chi}^{n \cdot \sigma \cdot \chi} \epsilon_i\bigl(g(y)\bigr) \cdot e^{- \frac{i}{\delta} \tau} = \sum\limits_{k = 0}^{+\infty} \tau^k \left(\sum\limits_{i = - n \cdot \sigma \cdot \chi}^{n \cdot \sigma \cdot \chi} \epsilon_i\bigl(g(y)\bigr) \cdot \frac{(-i)^k}{\delta^k k!}\right),\\
    \frac{1}{\bigl(1 - e^{-\beta_1 \tau}\bigr) \dots \bigl(1 - e^{-\beta_n \tau}\bigr)} = 
 \frac{1}{\tau^n \beta_1 \dots \beta_n} \sum\limits_{k = 0}^{+\infty} \tau^{k} \cdot \toddp_k(\beta_1, \dots, \beta_n),
\end{gather*} 
where the last formula could be taken, for example, from \cite[Chapter~14]{BarvBook}, and $\toddp_j(\beta_1,\dots,\beta_n)$ is a homogeneous polynomial of degree $j$, called the \emph{$j$-th Todd's polynomial} on $\beta_1,\dots,\beta_n$. Consequently,
\begin{multline*}
    \frac{\sum\limits_{i = - n \cdot \sigma\cdot \chi}^{n \cdot \sigma\cdot \chi} \epsilon_i\bigl(g(y)\bigr) \cdot e^{-\frac{i}{\delta} \tau}}{\bigl(1 - e^{-\beta_1 \cdot \tau}\bigr) \dots \bigl(1 - e^{-\beta_n \cdot \tau}\bigr)} = \\ 
    = \frac{1}{\tau^n \beta_1 \dots \beta_n} \sum\limits_{k = 0}^{+\infty} \tau^k \left( \sum\limits_{j = 0}^k \sum\limits_{i = - n \cdot \sigma\cdot \chi}^{n \cdot \sigma\cdot \chi} \epsilon_i\bigl(g(y)\bigr) \frac{(-i)^j}{\delta^j j!} \toddp_{k-j}(\beta_1, \dots, \beta_n) \right) = \\
    = \frac{1}{\tau^n \beta_1 \dots \beta_n} \sum\limits_{k = 0}^{+\infty} \tau^k \cdot \hat \pi_k\bigl(g(y)\bigr),\quad\text{ denoting }\\  
    \hat \pi_k(g) := \sum\limits_{j = 0}^k \toddp_{k-j}(\beta_1, \dots, \beta_n) \cdot \left( \frac{1}{\delta^j j!} \sum\limits_{i = - n \cdot \sigma\cdot \chi}^{n \cdot \sigma\cdot \chi} \epsilon_i(g) (-i)^j \right).
\end{multline*}
Therefore, due to the formula \eqref{single_vertex_repr_eq}, the constant term in $f(y;\tau)$ can be expressed by the formula
\begin{equation}\label{constant_term_0eq}
    \frac{1}{\beta_1 \dots \beta_n} \sum\limits_{k = 0}^{n} \frac{\bigl\langle c, A_{\BC}^{-1} u(y) \bigr\rangle^k}{k!} \cdot \hat \pi_{n-k}\bigl(g(y)\bigr).
\end{equation}
Let us define 
$$
\pi_{k}(x) = \frac{1}{k!\, \beta_1 \dots \beta_n} \cdot \hat \pi_{n-k}\bigl(x \bmod S \cdot \ZZ^n \bigr).
$$
Clearly, $\pi_{k} \colon \ZZ^n \to \QQ_{\geq 0}$ is a periodic function with a period-matrix $S$, in other words $\pi_{k}(x + \BUnit \cdot S) = \pi_{k}(x)$, for any $x \in \ZZ^n$. Denoting $c_{\BC} = A_{\BC}^{-\top} c$, and since $\bigl\langle c, A_{\BC}^{-1} \TC(y) \bigr\rangle = \bigl\langle A_{\BC}^{-\top} c, \TC(y) \bigr\rangle$, the formula \eqref{constant_term_0eq} can be rewritten in the following way:
\begin{equation}\label{constant_term_eq}
\sum\limits_{k = 0}^{n} \pi_{k}\bigl(P \TC(y)\bigr) \cdot \bigl\langle c_{\BC}, \TC(y) \bigr\rangle^k,
\end{equation} which is a periodic step-polynomial of degree $n$ and length $n+1$.

Due to \cite[Chapter~14]{BarvBook}, $\QEnum(y)$ is equal to the constant term in the Taylor's expansion of $f\bigl(\PC_{y} ; \tau\bigr)$. Clearly, it can be represented as the sum of constant terms of the Tailor's expansions of $f_{\BC}(y;\tau)$, for $\BC \in \BS$. Consequently, for a fixed chamber $\QC \in \QS$, the resulting function $\QEnum(y)$ could be represented as a periodic step-polynomial of degree $n$ and length $(n+1) \cdot \mu$. The exact form is the same as it was proposed in the formula \eqref{counting_func_repr_eq}.

Let us again fix $\BC \in \BS$ and discuss the arithmetic cost to compute $f(y;\tau) := f_{\BC}(y;\tau)$. Due to Lemma \ref{quad_system_lm}, the coefficients $\{\beta_i\}$, for $i \in \intint n$, and the coefficients $\{\epsilon_i(g)\}$, for any $i \in \intint[-n\cdot\sigma\cdot\chi]{n\cdot\sigma\cdot\chi}$ and $g \in \ZZ^n / S \ZZ^n$, can be computed with $O\bigl( T_{SNF}(n) + \Delta \cdot n^2 \cdot \sigma \cdot \chi \bigr)$ operations. Since $\sigma \leq \Delta$, $\chi = \Omega(n)$, and, due to Storjohann \cite{SNFOptAlg}, $T_{SNF}(n) = O(n^3)$, the last bound becomes $O(\Delta^2 \cdot n^2 \cdot \chi)$. 

Due to \cite[Theorem~7.2.8, p.~137]{AlgebracILP}, the values of $\toddp_{k}(\beta_1,\dots,\beta_n)$, for any $k \in \intint n$, can be computed with an algorithm that is polynomial in $n$ and the bit-encoding length of $\beta_1,\dots,\beta_n$. Moreover, it follows from the theorem's proof that the arithmetic complexity can be bounded by $O(n^3)$.

Consider now the values of $\{\hat \pi_k(g)\}$, for any $k \in \intint n$ and $g \in \ZZ^n / S \ZZ^n$. Due to the definition of $\hat \pi_k$, for a fixed $g$, the sequence $\bigl\{\hat \pi_k(g)\bigr\}_{k \in \intint[0]{n}}$ can be interpreted as a convolution of two sequences. The first sequence is already computed as the values of Todd's polynomials. Clearly, the second sequence can be computed with $O(n^2 \cdot \sigma \cdot \chi)$ operations. Now, using the fast convolution algorithms, the sequence $\bigl\{\hat \pi_k(g)\bigr\}_{k \in \intint[0]{n}}$ can be computed with $O(n \cdot \log(n))$ operations. So, the total arithmetic cost to compute $\{\hat\pi_k(g)\}$, for all $k \in \intint n$ and $g \in \ZZ^n / S \ZZ^n$, is $O(n^2 \cdot \Delta^2 \cdot \chi)$.

Summarizing the analysis, the generating function $f_{\BC}(y;\tau)$, for a fixed $\BC \in \BS$, can be computed with $O(n^2 \cdot \Delta^2 \cdot \chi)$ operations. Since $|\BS| \leq \mu$, we need 
$$
O\bigl(T_{trng} + \mu \cdot n^2 \cdot \Delta^2 \cdot \chi \bigr)
$$ operations to compute $f(\PC_{y};\tau)$, where $T_{trng}$ is the total arithmetic complexity of all the triangulations, performed during our algorithm. As it was already discussed, we have $T_{trng} = O(\mu \cdot n^3)$. So, the last bound becomes $O(\mu \cdot n^2 \cdot \Delta^2 \cdot \chi)$.

Previously, we made the assumption that the vector $c \in \ZZ^n$ is chosen, such that $c^\top h \not= 0$, for any $h \in \EC$. Let us present an algorithm that generates a vector $c$ with a respectively small value of the parameter $\chi = \max\limits_{h \in \EC} \bigl\{\abs{c^\top h}\bigr\}$. The main idea is concentrated in the following Theorem \ref{all_non_zero_th}, due to \cite{Counting_FPT_Delta_corrected} (see also \cite{SparseILP_Gribanov}). 
\begin{theorem}[Theorem~2 of \cite{Counting_FPT_Delta_corrected}]\label{all_non_zero_th}
Let $\AC$ be a set composed of $N$ non-zero vectors in $\QQ^n$. Then, there exists a randomized algorithm with the expected arithmetic complexity $O(n \cdot N)$, which finds a vector $z \in \ZZ^n$, such that:
\begin{enumerate}
    \item $a^\top z \not= 0$, for any $a \in \AC$;
    \item $\|z\|_{\infty} \leq N$.
\end{enumerate}
\end{theorem}

Any element of the triangulation of $\cone(A^\top)$ generates at most $n$ edges of $\PC_y$. Consequently, $\abs{\EC} \leq \mu \cdot n$. Choose some base $\BC$ of $A$. Note that $A_{\BC} h \not= \BZero$ and $(A_{\BC} h)_i \in \intint[-\Delta]{\Delta}$, for any $h \in \EC$ and $i \in \intint n$. Next, we use Theorem \ref{all_non_zero_th} to the set $A_{\BC} \cdot \EC$, which produces a vector $z$, such that 
\begin{enumerate}
    \item $z^\top A_{\BC} h \not= 0$, for each $h \in \EC$;
    \item $\|z\|_{\infty} \leq \mu \cdot n$.
\end{enumerate}
Now, we assign $c := A_{\BC}^\top z$. By the construction, we have $c^\top h \not= 0$ and $\abs{c^\top h} = \abs{z^\top A_{\BC} h} \leq n^2 \cdot \mu \cdot \Delta$, for each $h \in \EC$. Consequently, $\chi \leq n^2 \cdot \mu \cdot \Delta$. Therefore, we can conclude that, for a fixed $\QC \in \QS$ and any $y \in \relint(\QC)$, the representation of $f(\PC_y; \tau)$, as an $n$-degree periodic step-polynomial of degree $n$ and length $(n+1) \cdot \mu$, can be found by an algorithm with the arithmetic complexity 
$$
O(\mu^2 \cdot n^4 \cdot \Delta^3 ).
$$

\subsection{How to store the data, and what is the final preprocessing time?}\label{totall_preproc_subs}

Due to Theorem \ref{chamber_decomp_th}, the arithmetic complexity to construct the collection of chambers $\QS$ together with their parametric vertices is $O^*\bigl((f_{n_y-1})^{n_y} \cdot f_{n_y} + (f_{n_y-1})^{2n_y} \cdot (f_{n_y-1} + \nu)\bigr)$. For a fixed chamber $\QC \in \QS$ and $y \in \relint(\QC)$, the counting function $\QEnum(y)$ is represented as the periodic step-polynomial, given by the formula \eqref{counting_func_repr_eq}. Its length is $O(\mu \cdot n_x)$, its degree is $n_x$, and it can be computed with $O(\mu^2 \cdot n_x^4 \cdot \Delta^3)$ operations. To store such a representation, for each $\BC \in \BS$, we need to store $A_{\BC}=P_{\BC} S_{\BC} Q_{\BC}$, $c_{\BC}$, and the $O(n_x \cdot \Delta)$ values of the periodic coefficients $\pi_{\BC,k}(g)$, for $k \in \intint[0]{n_x}$ and $g \in \ZZ^{n_x} \bmod\, S_{\BC} \cdot \ZZ^{n_x}$. For a fixed $k$, the values of $\{\pi_{\BC,k}(g)\}$ can be stored in a hash-table of the size $|\det(S)| \leq \Delta$, where the keys are vectors from $\ZZ^{n_x}$, whose $i$-th component is from $\intint[0]{S_{i i} -1}$. Since, due to Theorem \ref{chamber_decomp_th}, $|\QS| = O\bigl( (f_{n_y-1})^{2n_y} \bigr)$ and since $\nu \leq \mu$, the total preprocessing arithmetic complexity can be estimated by
\begin{equation*}
    O^*\bigl( (f_{n_y-1})^{n_y} \cdot f_{n_y} + (f_{n_y-1})^{2n_y} \cdot (f_{n_y - 1} + \mu^2 \cdot \Delta^3) \bigr).
\end{equation*}

\subsection{What is the query time?}\label{total_query_subs}

Let us estimate the complexity of an $\QEnum(y)$-query, for a given vector $y \in \QQ^{n_y}$. First, we need to find a chamber $\QC \in \QS$, such that $y \in \relint(\DC)$. Due to Theorem \ref{chamber_decomp_th}, it can be done with $O(n_y \cdot f_{n_y-1})$ operations.

As it was shown before, for a fixed $\QC \in \QS$ and any $y \in \relint(\DC)$, there exists a set of bases $\BS$, such that $\QEnum(y)$ equals to the sum of constant terms in Taylor's decompositions of $f_{\BC}(y;\tau)$, for $\BC \in \BS$. Recalling that, for $\BC \in \BS$, the objects $A_{\BC}=P_{\BC} S_{\BC} Q_{\BC}$, $c_{\BC}$, and $\pi_{\BC,k}$ are already pre-computed, let us show how to compute the corresponding constant term: 
\begin{enumerate}
    \item Compute $\TC_{\BC}(y) = \bigl\lfloor b_{\BC} - B_{\BC} y \bigr\rfloor$ with $O(n_x \cdot n_y)$ operations;
    \item Compute $g = P_{\BC} \TC_{\BC}(y) \bmod S_{\BC} \cdot \ZZ^{n_x}$ with $O(n_x \cdot \log(\Delta))$ operations;
    \item For $k \in \intint[0]{n_x}$, access the coefficients $\pi_{\BC,k}\bigl(g\bigr)$, using the corresponding hash-table. It takes $O(n_x \cdot \log(\Delta))$ operations;
    \item Compute $\bigl\langle c_{\BC}, \TC_{\BC}(y) \bigr\rangle$ with $O(n_x)$ operations;
    \item For each $k \in \intint[0]{n_x}$, compute $\bigl\langle c_{\BC}, \TC_{\BC}(y) \bigr\rangle^k$ with $O(n_x)$ operations;
    \item Compute the formula \eqref{constant_term_eq} with $O(n_x)$ operations.
\end{enumerate}

After that, we just need to take the sum along all the constant terms, corresponding to $\BC \in \BS$. As it was noted before, there are at most $\mu$ terms. The total arithmetic cost is
$$
O\Bigl( n_y \cdot f_{n_y - 1} +  \mu \cdot n_x \cdot \bigl(\log(\Delta) + n_y\bigr) \Bigr).
$$

\addcontentsline{toc}{section}{Acknowledgement}
\section*{Acknowledgement}
Main results of our work obtained in Sections \ref{new_repr_sec} and \ref{main_param_th_proof} were supported by the Russian Science Foundation (RScF) No. 24-71-10021. The results obtained in Sections \ref{special_sec} and \ref{chamber_decomp_proof} 
were prepared within the framework of the Basic Research Program at the National Research University Higher School of Economics (HSE).
% This work was supported by a grant for research centers in the field of artificial intelligence, provided by the Analytical Center for the Government of the Russian Federation in accordance with the subsidy agreement (agreement identifier 000000D730324P540002) and the agreement with the Moscow Institute of Physics and Technology dated November 1, 2021 No. 70-2021-00138.

\addcontentsline{toc}{section}{Statements and declarations}
\section*{Statements and declarations}

\noindent{\bf Competing interests:} The authors have no competing interests.

\noindent{\bf Data availability statement:} The manuscript has no associated data.

\begin{appendices}

\section{Proof o Lemma \ref{quad_system_lm}}\label{quad_system_th_proof}
\begin{proof}

After the unimodular map $x = Q x'$ and introducing slack variables $y$, the system $\{x \in \ZZ^n \colon A x \leq b\}$ transforms to 
$$
\begin{cases}
S x + P y = P b\\
x \in \ZZ^{n}\\
y \in \ZZ^{n}_{\geq 0}.
\end{cases}
$$
Due to unimodularity of $P$, the last system is equivalent to
\begin{equation}\label{initial_group_system_eq}
\begin{cases}
P y = P b \pmod{S \cdot \ZZ^n}\\
y \in \ZZ^{n}_{\geq 0}.
\end{cases}
\end{equation}
Denoting $\GC = \ZZ^{n}/S \cdot \ZZ^n$, $g_0 = P b \bmod S \cdot \ZZ^n$, $g_i = P_{* i} \bmod S \cdot \ZZ^n$, we rewrite the last system \eqref{initial_group_system_eq}:
\begin{equation}\label{group_system_eq}
\begin{cases}
\sum\limits_{i = 1}^n y_i g_i = g_0\\
y \in \ZZ_{\geq 0}^n.
\end{cases}    
\end{equation}
The points $x \in \PC \cap \ZZ^n$ and the solutions $y$ of the system \eqref{group_system_eq} are connected by the bijective map $x = A^{-1}(b - y)$. Let $r_i = \abs{\langle g_i \rangle}$, for $i \in \intint{n}$, and $r_{\max} := \max_{i \in \intint n} \{r_i\}$. Clearly, $\abs{\GC} = \abs{\det(S)} = \Delta$ and $r_{\max} \leq \sigma$. For $k \in \intint{n}$ and $g' \in \GC$, let $\MC_k(g')$ be the solutions set of the auxiliary system $$
\begin{cases}
\sum\limits_{i = 1}^k y_i g_i = g'\\
y \in \ZZ_{\geq 0}^k,
\end{cases}
$$ and denote
$$
\gG_k(g'; \tau) = \sum\limits_{y \in \MC_k(g')} e^{-\langle c, \sum\limits_{i=1}^k h_{i} y_i 
\rangle \tau}.
$$
It follows that
\begin{multline}\label{group_connection_eq}
    \fG(\PC,c;\,\tau) = \sum\limits_{z \in \PC \cap \ZZ^{n}} e^{\langle c, z 
 \rangle \tau}
    = \sum\limits_{y \in \MC_n(g_0)} e^{\langle c, A^{-1}(b-y) \rangle \tau} = \\
    = e^{ \langle c, A^{-1} b \rangle \tau} \cdot \sum\limits_{y \in \MC_n(g_0)} e^{-\frac{1}{\Delta} \langle c, A^* y \rangle \tau} =  e^{ \langle c, A^{-1} b \rangle \tau} \cdot \gG_n\bigl(g_0;\frac{\tau}{\Delta}\bigr).
\end{multline}
The recurrent formulae for $\gG_k(g';\tau)$ were formally proven in \cite[see its formulae (10), (11), and (12)]{Counting_FPT_Delta}, we cite them using the following separate Lemma \ref{g_formulae_lm}. The self-contaned proof of the Lemma is given in Section \ref{g_formulae_lm_proof}.
% Since the original published paper \cite{Counting_FPT_Delta} contained an inaccuracy in the main result, we give a self-contained proof of the lemma in Subsection \ref{g_formulae_proof} of Appendix.
\begin{lemma}\label{g_formulae_lm}
The following formulae hold:
    \begin{gather}
    \gG_1(g'; \tau) = \frac{e^{- \langle c, s  h_1 \rangle 
 \tau}}{1 - e^{- \langle c, r_1  h_1 \rangle \tau}},\quad\text{where $s = \min\{y_1 \in \ZZ_{\geq 0} \colon y_1 \cdot g_1 = g' \}$},\label{gg_k_tau_initial}\\
    \gG_k(g';\tau) = \frac{1}{1 - e^{-\langle c, r_k  h_k \rangle \tau}} \cdot \sum\limits_{i = 0}^{r_k-1} e^{- \langle c, i h_k \rangle \tau} \cdot \gG_{k-1}(g' - i \cdot g_k; \tau),\label{gg_k_tau_recur}\\
    \gG_k(g';\tau) = \frac{\sum\limits_{i = - k \cdot \sigma \cdot \psi}^{k \cdot \sigma \cdot \psi} \epsilon_i(k,g') \cdot e^{- i \tau}}{\bigl(1 - e^{-\langle c, r_1 \cdot h_1 \rangle \tau}\bigr)\bigl(1 - e^{-\langle c, r_2 h_2 \rangle \tau}\bigr) \dots \bigl(1 - e^{- \langle c, r_k h_k \rangle \tau}\bigr)},\label{gg_k_tau_conv}
\end{gather}
where $\epsilon_{i}(k,g') \in \ZZ_{\geq 0}$ are coefficients, depending on $k$ and $g'$. If the set $\{y_1 \in \ZZ_{\geq 0} \colon y_1 g_1 = g' \}$ is empty, we put $\gG_1(g'; \tau) := 0$. If the vector $c$ is chosen such that $\langle c, h_i \rangle > 0$, for all $i \in \intint n$, then, for any $\tau > 0$, $k \in \intint n$, and $g' \in \GC$, the series $\gG_k(g'; \tau)$ converges absolutely to the corresponding r.h.s.\, functions.
\end{lemma} 

Let us estimate the number of operations to compute the representation \eqref{gg_k_tau_conv} of $\gG_k(g';\tau)$, for all $k \in \intint n$ and $g' \in \GC$, using the recurrence \eqref{gg_k_tau_recur}. Consider a quotient group $\QS_k = \GC/\langle g_k \rangle$ and fix $\QC \in \QS_k$. Clearly, $\QC = q + \langle g_k \rangle$, where $q \in \GC$ is a member of $\QC$, and $r_k = \abs{\QC}$. For $j \in \intint[0]{r_k-1}$, define
\begin{equation}\label{hh_def}
\hG_k(j;\tau) = \bigl(1 - e^{-\langle c, r_1 h_1 \rangle \tau}\bigr) \cdot \ldots \cdot \bigl(1 - e^{- \langle c, r_k h_k \rangle \tau}\bigr) \cdot \gG_k(q + j \cdot g_k;\tau).    
\end{equation}
For the sake of simplicity, denote $x \ominus_{k} y = (x - y) \bmod r_k$, then the formulas \eqref{gg_k_tau_initial}, \eqref{gg_k_tau_recur} and \eqref{gg_k_tau_conv} become
\begin{gather}
     \hG_1(j; \tau) = e^{- \langle c, s  h_1 \rangle \tau},\quad\text{where $s = \min\{y_1 \in \ZZ_{\geq 0} \colon y_1 g_1 = q + j \cdot g_1 \}$},\label{hh_k_tau_initial}\\
    \hG_k(j;\tau) = \sum\limits_{i = 0}^{r_k-1} e^{- \langle c, i h_k \rangle \tau} \cdot \hG_{k-1}\bigl(j \ominus_{k} i; \tau\bigr),\label{hh_k_tau_recur}\\
    \hG_k(j;\tau) = \sum\limits_{i = - k \cdot \sigma \cdot \psi}^{k \cdot \sigma \cdot \psi} \epsilon_i(k,q + j \cdot g_k) \cdot e^{- i \tau}.\label{hh_k_tau_conv}
\end{gather}
%If the equation $\{y_1 \in \ZZ_{\geq 0} \colon y_1 g_1 = q + j\cdot g_1 \}$ has no solutions, it holds $\hG_1(j; \tau) = 0$. 
First, assume that $k=1$. Then, clearly, all the values $$\hG_1(0;\tau), \hG_1(1;\tau), \dots, \hG_1(r_1-1;\tau)$$ can be computed with $O(r_1)$ operations.
Assume now that $k \geq 2$ and that $(k-1)$-th level has already been computed. By the $k$-th level, we mean all the functions $\hG_k(j;\tau)$, for $j \in \intint[0]{r_k-1}$.
Due to the formula \eqref{hh_k_tau_conv}, $\hG_k(j;\tau)$ contains $O(k \cdot \sigma \cdot \psi)$ monomials. Thus, the function $\hG_k(0;\tau)$ can be computed directly using the formula \eqref{hh_k_tau_recur} with $O(r_k \cdot k \cdot \sigma \cdot \psi)$ operations. For $j \geq 1$, we have
\begin{multline}\label{hh_k_tau_smart_recur}
    \hG_k(j;\tau) = \sum\limits_{i = 0}^{r_k-1} e^{- \langle c, i h_k \rangle \tau} \cdot \hG_{k-1}(j \ominus_k i; \tau) =\\
    = \sum\limits_{i = -1}^{r_k-2} e^{- \langle c, (i+1) h_k \rangle \tau} \cdot \hG_{k-1}\bigl( j \ominus_k (i+1) ; \tau\bigr) =\\
    = e^{- \langle c, h_k \rangle \tau} \cdot \hG_k(j-1;\tau) +\hG_{k-1}(j;\tau) - e^{- \langle c, r_k h_k \rangle \tau} \cdot \hG_{k-1}\bigl(j \ominus_k r_k;\tau\bigr) = \\
    = e^{- \langle c, h_k \rangle \tau} \cdot \hG_k(j-1;\tau) + (1 - e^{- \langle c, r_k h_k \rangle \tau}) \cdot \hG_{k-1}\bigl(j;\tau\bigr).
\end{multline}
Consequently, due to the assumption that the $(k-1)$-th level has already been computed and that $h_k(0;\tau)$ is known, all the functions $h_k(1;\tau), \dots, h_k(r_k-1;\tau)$ can be computed with $O(r_k \cdot k \cdot \sigma \cdot \psi)$ operations, using this formula \eqref{hh_k_tau_smart_recur}. 

In turn, when the functions $\hG_k(j;\tau)$, for $j \in \intint[0]{r_k-1}$, are already constructed, we can return to the functions $\gG_k(g';\tau)$, for $g' = q + j \cdot g_k$, using the formula \eqref{hh_def}. It will consume additional $O(r_k)$ group operations to compute $g' = q + j \cdot g_k$. By the definition of $\GC$, the arithmetic cost of a single group operation in $\GC$ can be estimated by the number of elements on the diagonal of the matrix $S$ that are not equal to $1$, which is bounded by $\min\{n, \log_2(\Delta)\}$. Therefore, the arithmetic cost of the last step is $O(r_k \cdot n)$, which is negligible in comparison with the computational cost of $\hG_k(j;\tau)$. 

Summarizing, we need $O(r_k \cdot k \cdot \sigma \cdot \psi)$ arithmetic operations to compute $\gG_k(g';\tau)$, for all $g' = q + j \cdot g_k$ and $j \in \intint[0]{r_k}$. Consequently, since $\abs{\QS} = \Delta/r_k$, the arithmetic cost to compute $k$-th level of $\gG_k(\cdot)$ is $$
O(\Delta \cdot k \cdot \sigma \cdot \psi),
$$ and the total cost to compute all the levels is $$
O(\Delta \cdot n^2 \cdot \sigma \cdot \psi).
$$
Finally, using the formula \eqref{group_connection_eq}, we construct the desired function
\begin{multline*}
    \fG(\PC,c;\tau) = e^{\langle c, A^{-1} b \rangle \tau} \cdot \gG_n\bigl(g_0; \frac{\tau}{\Delta}\bigr) = \\
    = \frac{\sum\limits_{i = - k \cdot \sigma \cdot \psi}^{k \cdot \sigma \cdot \psi} \epsilon_i \cdot e^{\frac{1}{\Delta}\bigl(\langle c, A^* b \rangle - i\bigr) \tau}}{\bigl(1 - e^{-\langle c, \frac{r_1}{\Delta} h_1 \rangle \tau}\bigr)\bigl(1 - e^{-\langle c, \frac{r_2}{\Delta} h_2 \rangle \tau}\bigr) \dots \bigl(1 - e^{- \langle c, \frac{r_n}{\Delta} h_n \rangle \tau}\bigr)},
\end{multline*}
where $\epsilon_i:= \epsilon_i(n,g_0)$. Since, for all $\tau >0$, the series $\gG_n(g_0;\tau)$ converges absolutely, the same is true for $\fG(\PC,c;\tau)$. The number of operations to make the last transformation is proportional to the nominator length of $\gG_n(g_0;\tau)$, which is $O(n \cdot \sigma \cdot \psi)$.
\end{proof}

\section{Proof of Lemma \ref{g_formulae_lm}}\label{g_formulae_lm_proof}

\subsection{A Recurrent Formula for the Generating Function of a Group Polyhedron}

Let $\GC$ be an arbitrary finite Abelian group and $g_1,\dots,g_n \in \GC$. Let additionally $r_i = \abs{\langle g_i \rangle}$ be the order of $g_i$, for $i \in \intint n$, and $r_{\max} = \max_{i} \{r_i\}$. For $g' \in \GC$ and $k \in \intint n$, let $\MC(k,g')$ be the solutions set of the following system:
\begin{equation}\label{f_k_system}
    \begin{cases}
    \sum\limits_{i = 1}^k x_i g_i = g'\\
    x \in \ZZ_{\geq 0}^k.
    \end{cases}
\end{equation}
Consider the formal series 
$
\fG_k(g';\xB) = \sum\limits_{z \in \MC(k,g') \cap \ZZ^k} \xB^z.
$
For $k = 1$, we have
\begin{equation}\label{f_k_1form}
\fG_1(g';\xB) = \frac{x_1^s}{1 - x_1^{r_1}},\quad\text{where $s = \min\{x_1 \in \ZZ_{\geq 0} \colon x_1 g_1 = g'\}$.}    
\end{equation}
If such $s$ does not exist, we set $\fG_1(g';\xB) := 0$. Clearly, the series $\fG_1(g';\xB)$ absolutely converges for any $x_1 \in \CC$ with $\abs{x_1^{r_1}} < 1$. For any fixed $x_k \in \ZZ_{\geq 0}$, the system \eqref{f_k_system} can be rewritten as
\begin{equation*}
    \begin{cases}
    \sum\limits_{i = 1}^{k-1} x_i g_i = g' - x_k g_k\\
    x \in \ZZ_{\geq 0}^{k-1}.
    \end{cases}
\end{equation*}
Thus, for any $k \geq 1$,
\begin{multline}\label{f_k_recurrence}
    \fG_k(g';\xB) = \\
    = \frac{ \fG_{k-1}(g';\xB) + x_{k} \cdot \fG_{k-1}(g' - g_k;\xB) + \dots + x_{k}^{r_k - 1} \cdot \fG_{k-1}(g' - g_k \cdot (r_k - 1);\xB)} {1 - x_k^{r_k}} = \\
    = \frac{1}{1 - x_{k}^{r_k}} \cdot \sum_{i = 0}^{r_k - 1} x_k^i \cdot \fG_{k-1}(g' - i \cdot g_k;\xB).
\end{multline}

\begin{equation}\label{f_k_conv}
\text{Consequently,} \quad \fG_k(g';\xB) = \frac{\sum\limits_{i_1 = 0}^{r_1-1}\dots\sum\limits_{i_k = 0}^{r_k-1} \epsilon_{i_1,\dots,i_k} x_1^{i_1} \dots x_k^{i_k}}{(1 - x_1^{r_1})(1 - x_2^{r_2})\dots(1 - x_k^{r_k})},    
\end{equation}
where the numerator is a polynomial with coefficients $\epsilon_{i_1,\dots,i_k} \in \{0,1\}$ of a degree at most $(r_1 - 1) \dots (r_k - 1)$. Since a sum of absolutely convergent series is absolutely convergent, it follows from the induction principle that the series $\fG_k(g';\xB)$ absolutely converges when $\abs{x_i^{r_i}} < 1$, for each $i \in \intint k$.

\subsection{The group $\GC$, induced by the SNF, of $A$}
Recall that $A \in \ZZ^{n \times n}$, $0 < \Delta = \abs{\det(A)}$, and $h_1, \dots, h_n$ are the columns of $A^* := \Delta \cdot A^{-1}$. The vector $c \in \ZZ^n$ is chosen, such that $\langle c, h_i \rangle > 0$, for each $i \in \intint n$, and $\psi = \max_i \abs{ \langle c, h_i \rangle }$.  Additionally, $S = P A Q$ is the SNF of $A$, where $P,Q \in \ZZ^{n \times n}$ are unimodular, and $\sigma = S_{n n}$.

Consider the sets $\MC(k,g')$, induced by the group system \eqref{f_k_system} with $\GC = \ZZ^{n}/S \cdot \ZZ^n$ and $g_i = P_{* i} \bmod S \cdot \ZZ^n$. Note that $r_i \leq \sigma$, for each $i \in \intint n$. Additionally, consider a new formal series 
$$
\hat \fG_k(g';\xB) = \sum\limits_{z \in \MC(k,g') \cap \ZZ^k} \xB^{-\sum\limits_{i=1}^k h_i z_i},
$$ which can be derived from the series $\fG_k(g';\xB)$ by the monomial substitution $x_i \to \xB^{-h_i}$.
For $\hat \fG_k(g';\xB)$, the formulae \eqref{f_k_1form}, \eqref{f_k_recurrence} and \eqref{f_k_conv} become:
\begin{gather}
    \hat \fG_1(g'; \xB) = \frac{\xB^{- s h_1}}{1 - \xB^{-r_1 h_1}},\quad\text{where $s = \min\{y_1 \in \ZZ_{\geq 0} \colon y_1 g_1 = g' \}$,}\label{ff_k_1form}\\
    \hat \fG_k(g';\xB) = \frac{1}{1 - \xB^{-r_k h_k}} \cdot \sum\limits\limits_{i = 0}^{r_k-1} \xB^{- i h_k} \cdot \hat \fG_{k-1}(g' - i \cdot g_k; \xB) \text{ and}\label{ff_k_recur}\\    
\hat \fG_k(g';\xB) = \frac{\sum\limits_{i_1 = 0}^{r_1-1}\dots\sum\limits_{i_k = 0}^{r_k-1} \epsilon_{i_1,\dots,i_k} \xB^{-(i_1 h_1 + \dots + i_k h_k)}}{(1 - \xB^{-r_1 h_1})(1 - \xB^{-r_2 h_2}) \dots (1 - \xB^{-r_k h_k})}\label{ff_k_conv}.
\end{gather}
Here the absolute convergence takes place for the values of $\xB$ with $\abs{\xB^{- r_i h_i}} < 1$, for each $i \in \intint k$. Let us consider now the formal series 
$$
\gG_k(g'; \tau) = \sum\limits_{y \in \MC_k(g')} e^{- \tau \cdot \langle c, \sum_{i=1}^k h_i y_i \rangle},
$$ which can be derived from $\hat \fG_k(g';\xB)$ substituting $x_i \to e^{\tau \cdot c_i}$. For $\gG_k(g'; \tau)$, the formulae \eqref{ff_k_1form}, \eqref{ff_k_recur}, and \eqref{ff_k_conv} become:
\begin{gather}
    \gG_1(g'; \tau) = \frac{e^{- \langle c, s  h_1 \rangle \cdot \tau}}{1 - e^{- \langle c, r_1  h_1 \rangle \cdot \tau}},\label{g_k_1form}\\
    \gG_k(g';\tau) = \frac{1}{1 - e^{-\langle c, r_k  h_k \rangle \cdot \tau}} \cdot \sum\limits_{i = 0}^{r_k-1} e^{- \langle c, i h_k \rangle \cdot \tau} \cdot \gG_{k-1}(g' - i \cdot g_k; \tau),\label{g_k_recur}\\
    \gG_k(g';\tau) = \frac{\sum\limits_{i_1 = 0}^{r_1-1}\dots\sum\limits_{i_k = 0}^{r_k-1} \epsilon_{i_1,\dots,i_k} e^{-\langle c, i_1 h_1 + \dots + i_k h_k \rangle \cdot \tau} }{\bigl(1 - e^{-\langle c, r_1 h_1 \rangle \cdot \tau}\bigr)\bigl(1 - e^{-\langle c, r_2 h_2 \rangle \cdot \tau}\bigr) \dots \bigl(1 - e^{- \langle c, r_k h_k \rangle \cdot \tau}\bigr)}.\label{g_k_conv}
\end{gather}
Since the series $\hat \fG_k(g';\xB)$ absolutely converges, when $\abs{\xB^{- r_i h_i}} < 1$, and since $\langle c,h_i \rangle \in \ZZ_{\not=0}$, for each $i \in \intint k$, the new series converges for any $\tau > 0$. The number of terms $e^{-\langle c, \cdot \rangle \cdot \tau}$ is bounded by $2 \cdot k \cdot \sigma\cdot \psi + 1$. Combining similar terms, the numerator's length becomes $O(k \cdot \sigma\cdot \psi)$.
In other words, there exist coefficients $\epsilon_i \in \ZZ_{\geq 0}$, such that 
\begin{equation}\label{g_k_coef_form}
\gG_k(g';\tau) = \frac{\sum\limits_{i = - k \cdot \sigma \cdot \psi}^{k \cdot \sigma \cdot \psi} \epsilon_i \cdot e^{- i \cdot \tau}}{\bigl(1 - e^{-\langle c, r_1 \cdot h_1 \rangle \tau}\bigr)\bigl(1 - e^{-\langle c, r_2 h_2 \rangle \cdot \tau}\bigr) \dots \bigl(1 - e^{- \langle c, r_k h_k \rangle \cdot \tau}\bigr)}.    
\end{equation}
The formulae \eqref{g_k_1form}, \eqref{g_k_recur}, and \eqref{g_k_coef_form} coincide with the desired formulae \eqref{gg_k_tau_initial}, \eqref{gg_k_tau_recur}, and \eqref{gg_k_tau_conv}. Therefore, the proof of Lemma \ref{g_formulae_lm} is finished.

\end{appendices}

\addcontentsline{toc}{section}{References}
\bibliography{grib_biblio}

\end{document}

%% file: grib_defs.tex
\DeclareMathOperator{\NN}{\mathbb{N}}

\DeclareMathOperator{\CC}{\mathbb{C}}

\DeclareMathOperator{\RR}{\mathbb{R}}

\DeclareMathOperator{\QQ}{\mathbb{Q}}

\DeclareMathOperator{\ZZ}{\mathbb{Z}}

\DeclareMathOperator{\GC}{\mathcal{G}}
\DeclareMathOperator{\RC}{\mathcal{R}}
\DeclareMathOperator{\BC}{\mathcal{B}}

\DeclareMathOperator{\VC}{\mathcal{V}}
\DeclareMathOperator{\WC}{\mathcal{W}}
\DeclareMathOperator{\TC}{\mathcal{T}}
\DeclareMathOperator{\EC}{\mathcal{E}}
\DeclareMathOperator{\PC}{\mathcal{P}}
\DeclareMathOperator{\DC}{\mathcal{D}}
\DeclareMathOperator{\SC}{\mathcal{S}}

\DeclareMathOperator{\LC}{\mathcal{L}}
\DeclareMathOperator{\HC}{\mathcal{H}}
\DeclareMathOperator{\MC}{\mathcal{M}}
\DeclareMathOperator{\AC}{\mathcal{A}}
\DeclareMathOperator{\FC}{\mathcal{F}}

\DeclareMathOperator{\JC}{\mathcal{J}}
\DeclareMathOperator{\IC}{\mathcal{I}}
\DeclareMathOperator{\CCal}{\mathcal{C}}
\DeclareMathOperator{\QC}{\mathcal{Q}}

\DeclareMathOperator{\BS}{\mathscr{B}}

\DeclareMathOperator{\DS}{\mathscr{D}}

\DeclareMathOperator{\HS}{\mathscr{H}}
\DeclareMathOperator{\RS}{\mathscr{R}}

\DeclareMathOperator{\PS}{\mathscr{P}}

\DeclareMathOperator{\QS}{\mathscr{Q}}
\DeclareMathOperator{\TS}{\mathscr{T}}

\DeclareMathOperator{\fG}{\mathfrak{f}}
\DeclareMathOperator{\gG}{\mathfrak{g}}
\DeclareMathOperator{\hG}{\mathfrak{h}}

\DeclareMathOperator{\rank}{rank}
\DeclareMathOperator{\cone}{cone}

\DeclareMathOperator{\affh}{aff.\!hull}
\DeclareMathOperator{\linh}{span}
\DeclareMathOperator{\inth}{\Lambda}

\DeclareMathOperator{\vertex}{vert}
\DeclareMathOperator{\pvertex}{pvert}
\DeclareMathOperator{\base}{Bases}
\DeclareMathOperator{\lcm}{lcm}

\DeclareMathOperator{\poly}{poly}

\DeclareMathOperator{\relint}{rel.\!int}
\DeclareMathOperator{\inter}{int}

\DeclareMathOperator{\diag}{diag}

\DeclareMathOperator{\tcone}{tcone}
\DeclareMathOperator{\fcone}{fcone}
\DeclareMathOperator{\lmod}{\text{\it\quad modulo polyhedra with lines}}

\DeclareMathOperator{\ldcmod}{\text{\it\quad modulo lower-dimensional rational cones}}
\DeclareMathOperator{\toddp}{td}

\DeclareMathOperator{\denom}{den}
\DeclareMathOperator{\chdenom}{ch.\!den}

\DeclareMathOperator{\BUnit}{\mathbf 1}
\DeclareMathOperator{\BZero}{\mathbf 0}
\DeclareMathOperator{\xB}{\mathbf{x}}

\DeclareMathOperator{\zB}{\mathbf{z}}

\DeclareMathOperator{\jB}{\mathbf{j}}

\DeclareMathOperator{\SNF}{\text{\rm SNF}}

\newcommand*{\intint}[2][1]{\{#1, \dots, #2\}}

\newcommand\restr[2]{{% we make the whole thing an ordinary symbol
  \left.\kern-\nulldelimiterspace % automatically resize the bar with \right
  #1 % the function
  \vphantom{\big|} % pretend it's a little taller at normal size
  \right|_{#2} % this is the delimiter
  }}